\input ./style/arxiv-general.cfg
\documentclass[MSNbibl,number,citesort,dvips]{arxbj}
\makeatletter
   \@ifpackageloaded{graphicx}{}{\usepackage{graphicx}}
\makeatother
\usepackage{upgreek}
\usepackage{float}

\volume{22}
\issue{1}
\pubyear{2016}
\firstpage{494}
\lastpage{529}
\doi{10.3150/14-BEJ666} 
\docsubty{FLA}

\makeatletter

\newcommand{\rrvert}{\vert}

\newcommand{\rrVert}{\Vert}
\newcommand{\llvert}{\vert}
\newcommand{\llVert}{\Vert}
\floatstyle{ruled}
\newfloat{algorithm}{tbp}{loa}
\newcommand{\algorithmname}{Algorithm}
\floatname{algorithm}{\protect\algorithmname}
  \newproclaim{assumption}{Assumption}
  \newremark{rem}{Remark}
  \newtheorem{lem}{Lemma}
  \newtheorem{cor}{Corollary}
\newtheorem{thmm}{Theorem}
  \newtheorem{prop}{Proposition}
 \newproclaim{defn}{Definition}

\makeatother

\begin{document}
\begin{frontmatter}

\title{On the role of interaction in sequential Monte Carlo algorithms}
\runtitle{On the role of interaction in SMC algorithms}

\begin{aug}
\author[A]{\inits{N.}\fnms{Nick}~\snm{Whiteley}\corref{}\thanksref
{A,e1}\ead[label=e1,mark]{nick.whiteley@bristol.ac.uk}},
\author[B]{\inits{A.}\fnms{Anthony}~\snm{Lee}\thanksref{B,e2}\ead
[label=e2,mark]{anthony.lee@warwick.ac.uk}}
\and
\author[A]{\inits{K.}\fnms{Kari}~\snm{Heine}\thanksref{A,e3}\ead
[label=e3,mark]{kari.heine@bristol.ac.uk}}
\address[A]{School of Mathematics, University of Bristol, University
Walk, Bristol BS8 1TW, United Kingdom.\\ \printead{e1,e3}}
\address[B]{Department of Statistics, University of Warwick, Coventry
CV4 7AL, United Kingdom.\\ \printead{e2}}
\end{aug}

\received{\smonth{9} \syear{2013}}
\revised{\smonth{7} \syear{2014}}

%
\begin{abstract}
We introduce a general form of sequential Monte Carlo algorithm defined
in terms of a parameterized resampling mechanism. We find that a suitably
generalized notion of the Effective Sample Size (ESS), widely used
to monitor algorithm degeneracy, appears naturally in a study of its
convergence properties. We are then able to phrase sufficient conditions
for time-uniform convergence in terms of algorithmic control of the
ESS, in turn achievable by adaptively modulating the interaction between
particles. This leads us to suggest novel algorithms which are, in
senses to be made precise, provably stable and yet designed to avoid
the degree of interaction which hinders parallelization of standard
algorithms. As a byproduct, we prove time-uniform convergence of the
popular adaptive resampling particle filter.
\end{abstract}

%
\end{frontmatter}

\section{Introduction}\label{sec1}

At the frontier of computational statistics there is growing interest
in parallel implementation of Monte Carlo algorithms using multi-processor
and distributed architectures. However, the resampling step of sequential
Monte Carlo (SMC) methods \cite{gordon1993novel} (see \cite
{kunsch2013particle}
for a recent overview) which involves a degree of interaction between
simulated ``particles'', hinders their parallelization. So, whilst
multi-processor implementation offers some speed up for SMC, the potential
benefits of distributed computing are not fully realized \cite{lee2010utility}.

Performing resampling only occasionally, a technique originally suggested
for the somewhat different reason of variance reduction \cite{liu1995blind},
alleviates this problem to some extent, but the collective nature
of the resampling operation remains the computational bottleneck.
On the other hand, crude attempts to entirely do away with the resampling
step may result in unstable or even non-convergent algorithms. With
these issues in mind, we seek a better understanding of the relationship
between the interaction structure of SMC algorithms and theoretical
properties of the approximations they deliver. Our overall aim is
to address the following question:\vspace*{6pt}

\emph{To what extent can the degree of interaction between particles
be reduced, whilst ensuring provable stability of the algorithm?}\vspace*{6pt}

Our strategy is to introduce and study an unusually general type of
SMC algorithm featuring a parameterized resampling mechanism. This\vadjust{\goodbreak}
provides a flexible framework in which we are ultimately able to attach
meaning to \emph{degree of interaction} in terms of graph-theoretic
quantities. To address the matter of \emph{provable stability}, we
seek conditions under which the algorithm yields time-uniformly convergent
approximations of prediction filters, and approximations of marginal
likelihoods whose relative variance can be controlled at a linear-in-time
cost.

The general algorithm we study is defined in terms of a family of
Markov transition matrices, $\alpha$, and we refer to the algorithm
itself as $\alpha$SMC. We shall see that through particular choices
of $\alpha$ one obtains, as instances of $\alpha$SMC, well-known
algorithms including sequential importance sampling (SIS), the bootstrap
particle filter (BPF) and the adaptive resampling particle filter
(ARPF) in which resampling is triggered by monitoring some functional
criterion, such as the Effective Sample Size (ESS) \cite{liu1995blind}.

Although the ESS does not necessarily appear in the definition of
the general $\alpha$SMC algorithm, we find that it does appear quite
naturally from the inverse quadratic variation of certain martingale
sequences in its analysis. This allows us to make precise a sense
in which algorithmic control of the ESS can guarantee stability of
the algorithm. Our results apply immediately to the ARPF, but our
study has wider-reaching methodological consequences: in our framework
it becomes clear that the standard adaptive resampling strategy is
just one of many possible ways of algorithmically controlling the
ESS, and we can immediately suggest new, alternative algorithms which
are provably stable, but designed to avoid the type of complete interaction
which is inherent to the ARPF and which hinders its parallelization.
The structure of this paper and our main contributions are as follows.

Section~\ref{sec:aSMC} introduces the general algorithm, $\alpha$SMC.
We explain how it accommodates several standard algorithms as particular
cases and comment on some other existing SMC methods.

Section~\ref{sec:Martingale-approximations-and} presents Theorem~\ref{thmm:convergence}, a general convergence result for $\alpha$SMC.
We give conditions which ensure unbiased approximation of marginal
likelihoods and we elucidate connections between certain invariance
properties of the matrices $\alpha$ and the negligibility of increments
in a martingale error decomposition, thus formulating simple sufficient
conditions for weak and strong laws of large numbers. We also discuss
some related existing results.

Section~\ref{sec:stability} presents our second main result, Theorem~\ref{thmm:L_R_mix}. We show, subject to regularity conditions on the
hidden Markov model (HMM) under consideration, that enforcement of
a strictly positive lower bound on a certain coefficient associated
with ESS of $\alpha$SMC is sufficient to guarantee non-asymptotic,
time-uniform bounds on: (1) the exponentially normalized relative second
moment of error in approximation of marginal likelihoods, and (2) the
$L_{p}$ norm of error in approximation of prediction filters. The
former implies a linear-in-time variance bound and the latter implies
time-uniform convergence. These results apply immediately to the ARPF.

Section~\ref{sec:Discussion} houses discussion and application of
our results. We point out the pitfalls of some naive approaches to
parallelization of SMC and discuss what can go wrong if the conditions
of Theorem~\ref{thmm:convergence} are not met. Three new algorithms,
which adapt the degree of interaction in order to control the ESS
and which are therefore provably stable, are then introduced. We discuss
computational complexity and through numerical experiments examine
the degree of interaction involved in these algorithms and the quality
of the approximations they deliver compared to the ARPF.

\section{\texorpdfstring{$\alpha$}{alpha}SMC}
\label{sec:aSMC}

A hidden Markov model (HMM) with measurable state space $ (\mathsf
{X},\mathcal{X} )$
and observation space $ (\mathsf{Y},\mathcal{Y} )$ is a
process $ \{  (X_{n},Y_{n} );n\geq0 \} $ where
$ \{ X_{n};n\geq0 \} $ is a Markov chain on $\mathsf{X}$,
and each observation $Y_{n}$, valued in $\mathsf{Y}$, is conditionally
independent of the rest of the process given $X_{n}$. Let $\mu_{0}$
and $f$ be, respectively, a probability distribution and a Markov kernel
on $ (\mathsf{X},\mathcal{X} )$, and let $g$ be a Markov
kernel acting from $ (\mathsf{X},\mathcal{X} )$ to $
(\mathsf{Y},\mathcal{Y} )$,
with $g(x,\cdot)$ admitting a density, denoted similarly by $g(x,y)$,
with respect to some dominating $\sigma$-finite measure. The HMM
specified by $\mu_{0}$, $f$ and~$g$, is
%
\begin{eqnarray}\label{eq:HMM}
 X_{0}\sim\mu_{0}(\cdot),\qquad X_{n}|
\{X_{n-1}=x_{n-1}\}& \sim& f(x_{n-1},\cdot),\qquad n
\geq1,
\nonumber
\\[-8pt]
\\[-8pt]
\nonumber
 Y_{n}|\{ X_{n}=x_{n}\} &\sim&
g(x_{n},\cdot),\qquad n\geq 0.
\end{eqnarray}

We shall assume throughout that we are presented with a fixed observation
sequence $ \{ y_{n};n\geq0 \} $ and write
\[
g_{n}(x):=g(x,y_{n}),\qquad n\geq0.
\]
The following assumption imposes some mild regularity which ensures
that various objects appearing below are well defined. It shall be
assumed to hold throughout without further comment.

\begin{assumption*}
\textup{(A1)} For each $n\geq0$, $\sup_{x}g_{n}(x)<+\infty$
and $g_{n}(x)>0$ for all $x\in\mathsf{X}$.
\end{assumption*}

We take as a recursive definition of the \emph{prediction filters},
the sequence of distributions $ \{ \pi_{n};n\geq0 \} $ given
by
%
\begin{eqnarray}
\label{eq:filtering_recursion}  \pi_{0}&:=&\mu_{0},
\nonumber
\\[-8pt]
\\[-8pt]
\nonumber
\pi_{n} (A )&:=&\frac{\int_{\mathsf{X}}\pi_{n-1}
(\mathrm{d}x )g_{n-1}(x)f(x,A)}{\int_{\mathsf{X}}\pi_{n-1} (\mathrm{d}x
)g_{n-1}(x)},\qquad  A\in\mathcal{X}, n\geq1,
\end{eqnarray}
and let $ \{ Z_{n};n\geq0 \} $ be defined by
%
\begin{equation}
Z_{0}:=1, \qquad Z_{n}:=Z_{n-1}\int
_{\mathsf{X}}\pi_{n-1} (\mathrm{d}x )g_{n-1} (x ),\qquad n
\geq1.\label{eq:Z_recusion}
\end{equation}
Due to the conditional independence structure of the HMM, $\pi_{n}$
is the conditional distribution of $X_{n}$ given $Y_{0:n-1}=y_{0:n-1}$;
and $Z_{n}$ is the marginal likelihood of the first $n$ observations,
evaluated at the point $y_{0:n-1}$. Our main computational objectives
are to approximate $ \{ \pi_{n};n\geq0 \} $ and $ \{
Z_{n};n\geq0 \}$.

\subsection{The general algorithm}

With population size $N\geq1$, we write $[N]:=\{1,\ldots,N\}$. \emph{To
simplify presentation, whenever a summation sign appears without the
summation set made explicit, the summation set is taken to be $[N]$,
for example, we write $\sum_{i}$ to mean $\sum_{i=1}^{N}$.}

The $\alpha$SMC algorithm involves simulating a sequence $ \{ \zeta
_{n};n\geq0 \} $
with each $\zeta_{n}= \{ \zeta_{n}^{1},\ldots,\zeta_{n}^{N} \} $
valued in $\mathsf{X}^{N}$. Denoting $\mathbb{X}:= (\mathsf
{X}^{N} )^{\mathbb{N}}$,
$\mathcal{F}^{\mathbb{X}}:= (\mathcal{X}^{\otimes N} )^{\otimes
\mathbb{N}}$,
we shall view $ \{ \zeta_{n};n\geq0 \} $ as the canonical
coordinate process on the measurable space $ (\mathbb{X},\mathcal
{F}^{\mathbb{X}} )$,
and write $\mathcal{F}_{n}$ for the $\sigma$-algebra generated by
$ \{ \zeta_{0},\ldots,\zeta_{n} \} $. By convention, we
let $\mathcal{F}_{-1}:=\{\mathbb{X},\varnothing\}$ be the trivial $\sigma
$-algebra.
The sampling steps of the $\alpha$SMC algorithm, described below,
amount to specifying a probability measure, say $\mathbb{P}$, on
$ (\mathbb{X},\mathcal{F}^{\mathbb{X}} )$. Expectation
w.r.t. $\mathbb{P}$
shall be denoted by $\mathbb{E}$.

Let $\mathbb{A}_{N}$ be a non-empty set of Markov transition matrices,
each of size $N\times N$. For $n\geq0$ let $\alpha_{n}\dvtx \mathbb
{X}\rightarrow\mathbb{A}_{N}$
be a matrix-valued map, and write $\alpha_{n}^{ij}$ for the $i$th
row, $j$th column entry so that for each $i$ we have $\sum_{j}\alpha_{n}^{ij}=1$
(with dependence on the $\mathbb{X}$-valued argument suppressed).
The following assumption places a restriction on the relationship
between $\alpha$ and the particle system $ \{ \zeta_{n};n\geq0
\} $.
\begin{assumption*}
\textup{(A2)} For each $n\geq0$, the entries of $\alpha_{n}$
are all measurable with respect to $\mathcal{F}_{n}$
\end{assumption*}
Intuitively, the members of $\mathbb{A}_{N}$ will specify different
possible interaction structures for the particle algorithm and under
\textup{(A2)}, each $\alpha_{n}$ is a random matrix chosen from $\mathbb{A}_{N}$
according to some deterministic function of $ \{ \zeta_{0},\ldots,
\zeta_{n} \} $.
Examples are given below. We shall write $\mathbf{1}_{1/N}$ for the
$N\times N$ matrix which has $1/N$ as every entry and write $\mathrm{Id}$
for the identity matrix of size apparent from the context in which
this notation appears. We shall occasionally use $\mathrm{Id}$ also to denote
identity operators in certain function space settings. Let $\mathcal{M}$,
$\mathcal{P}$ and $\mathcal{L}$ be, respectively, the collections
of measures, probability measures and real-valued, bounded, $\mathcal
{X}$-measurable
functions on $\mathsf{X}$. We write
\[
\llVert \varphi\rrVert :=\sup_{x}\bigl\llvert \varphi(x)
\bigr\rrvert ,\qquad \operatorname{osc}(\varphi):=\sup_{x,y}\bigl
\llvert \varphi(x)-\varphi(y)\bigr\rrvert
\]
and
%
\begin{equation}
\mu(\varphi):=\int_{\mathsf{X}}\varphi(x)\mu(\mathrm{d}x) \qquad\mbox{for any }
\varphi\in\mathcal{L}, \mu\in\mathcal{M}.\label{eq:mu(phi)_notation}
\end{equation}

\begin{rem*}
Note that $\mathbb{X}$, $\mathcal{F}^{\mathbb{X}}$, $\mathcal{F}_{n}$,
$\mathbb{P}$, $\alpha$ and various other objects depend on $N$,
but this dependence is suppressed from the notation. Unless specified
otherwise, any conditions which we impose on such objects should be
understood as holding for all $N\geq1$.
\end{rem*}
Let $ \{ W_{n}^{i};i\in[N],n\geq0 \} $ be defined by the
following recursion:
%
\begin{equation}
W_{0}^{i}:=1, \qquad W_{n}^{i}:=\sum
_{j}\alpha _{n-1}^{ij}W_{n-1}^{j}g_{n-1}
\bigl(\zeta_{n-1}^{j}\bigr),\qquad i\in[N],n\geq 1.\label{eq:W_n_defn}
\end{equation}

The following algorithm implicitly specifies the law $\mathbb{P}$
of the $\alpha$SMC particle system. For each $n\geq1$, the ``Sample''
step should be understood as meaning that the variables $\zeta_{n}=
\{ \zeta_{n}^{i} \} _{i\in[N]}$
are conditionally independent given $ \{ \zeta_{0},\ldots,\zeta
_{n-1} \} $.
The line of Algorithm \ref{alg:aSMC} marked $(\star)$ is intentionally
generic, it amounts to a practical, if imprecise restatement of \textup{(A2)}.
In the sequel, we shall examine instances of $\alpha$SMC which arise
when we consider specific $\mathbb{A}_{N}$ and impose more structure
at line $(\star)$.

\begin{algorithm}[H]
\begin{raggedright}
\quad For $n=0$,
\end{raggedright}

\begin{raggedright}
\qquad For $i=1,\ldots,N$,
\end{raggedright}

\begin{raggedright}
\quad\qquad Set $W_{0}^{i}=1$
\end{raggedright}

\begin{raggedright}
\quad\qquad Sample $\zeta_{0}^{i}\sim\mu_{0}$
\end{raggedright}

\begin{raggedright}
\quad For $n\geq1$,
\end{raggedright}

\begin{raggedright}
$(\star)$\quad\hspace*{-4pt} Select $\alpha_{n-1}$
from $\mathbb{A}_{N}$ according to some functional of $ \{ \zeta
_{0},\ldots,\zeta_{n-1} \} $.
\end{raggedright}

\begin{raggedright}
\qquad For $i=1,\ldots,N$,
\end{raggedright}

\begin{raggedright}
\quad\qquad Set {} $W_{n}^{i}=\sum_{j}\alpha
_{n-1}^{ij}W_{n-1}^{j}g_{n-1}(\zeta_{n-1}^{j})$.
\end{raggedright}

\begin{raggedright}
\quad\qquad Sample $\zeta_{n}^{i}|\mathcal{F}_{n-1}
\sim \frac{\sum_{j}\alpha_{n-1}^{ij}W_{n-1}^{j}g_{n-1}(\zeta
_{n-1}^{j})f(\zeta_{n-1}^{j},\cdot)}{W_{n}^{i}}$.
\end{raggedright}
\caption{$\alpha$SMC}
\label{alg:aSMC}
\end{algorithm}

We shall study the objects
%
\begin{equation}
\pi_{n}^{N}:=\frac{\sum_{i}W_{n}^{i} \delta_{\zeta_{n}^{i}}}{\sum_{i}W_{n}^{i}},
\qquad
Z_{n}^{N}:=
\frac{1}{N}\sum_{i}W_{n}^{i},\qquad
n\geq0,\label{eq:pi^N_andZ^N}
\end{equation}
which as the notation suggests, are to be regarded as approximations
of $\pi_{n}$ and $Z_{n}$, respectively. We shall also be centrally
concerned with the following coefficient, which is closely related
to the ESS,
%
\begin{equation}
\mathcal{E}_{n}^{N}:=\frac{ (N^{-1}\sum_{i}W_{n}^{i}
)^{2}}{N^{-1}\sum_{i} (W_{n}^{i} )^{2}}=\frac{ (N^{-1}\sum_{i}\sum_{j}\alpha_{n-1}^{ij}W_{n-1}^{j}g_{n-1}(\zeta_{n-1}^{j})
)^{2}}{N^{-1}\sum_{i} (\sum_{j}\alpha
_{n-1}^{ij}W_{n-1}^{j}g_{n-1}(\zeta_{n-1}^{j}) )^{2}}, \qquad n
\geq 1,\label{eq:ESS_defn_front}
\end{equation}
and by convention $\mathcal{E}_{0}^{N}:=1$. The second equality in
(\ref{eq:ESS_defn_front}) is immediate from the definition of $W_{n}^{i}$,
see (\ref{eq:W_n_defn}). Note that $\mathcal{E}_{n}^{N}$ is always
valued in $[0,1]$, and if we write
%
\begin{equation}
N_{n}^{\mathrm{eff}}:=N\mathcal{E}_{n}^{N},\label{eq:N_eff}
\end{equation}
we obtain the ESS of Liu and Chen \cite{liu1995blind}, although of course in
a generalized form, since $\mathcal{E}_{n}^{N}$ is defined in terms
of the generic ingredients of $\alpha$SMC. A few comments on generality
are in order. First, for ease of presentation, we have chosen to
work with a particularly simple version of $\alpha$SMC, in which
new samples are proposed using the HMM Markov kernel $f$. The algorithm
is easily generalized to accommodate other proposal kernels. Second,
while we focus on the application of SMC methods to HMM's, our results
and methodological ideas are immediately transferable to other contexts,
for example, via the framework of \cite{smc:meth:DDJ06}.

\subsection{Instances of \texorpdfstring{$\alpha$}{alpha}SMC}
\label{sub:Instances-of-SMC}

We now show how $\alpha$SMC admits SIS, the BPF and the ARPF, as
special cases, through particular choices of $\mathbb{A}_{N}$. Our
presentation is intended to illustrate the structural generality of
$\alpha$SMC,\vadjust{\goodbreak} thus setting the scene for the developments which follow.
The following lemma facilitates exposition by ``unwinding'' the
quantities $ \{ W_{n}^{i} \} _{i\in[N]}$ defined recursively
in (\ref{eq:W_n_defn}). It is used throughout the remainder of the
paper.

\begin{lem}
\label{lem:W_n_representation}For $n\geq1$, $0\leq p<n$ and $i_{n}\in[N]$,
%
\begin{equation}
W_{n}^{i_{n}}=\sum_{ (i_{p},\ldots,i_{n-1} )\in
[N]^{n-p}}W_{p}^{i_{p}}
\prod_{q=p}^{n-1}g_{q}\bigl(
\zeta_{q}^{i_{q}}\bigr)\alpha _{q}^{i_{q+1}i_{q}},\label{eq:unwind}
\end{equation}
and in particular
%
\begin{equation}
W_{n}^{i_{n}}=\sum_{ (i_{0},\ldots,i_{n-1} )\in[N]^{n}}\prod
_{p=0}^{n-1}g_{p}\bigl(
\zeta_{p}^{i_{p}}\bigr)\alpha_{p}^{i_{p+1}i_{p}}.\label{eq:unwind2}
\end{equation}
\end{lem}

The proof of (\ref{eq:unwind})--(\ref{eq:unwind2}) is a simple induction
and is therefore omitted. From (\ref{eq:unwind2}) and definitions
above, we immediately observe the following corollary.

\begin{cor}
\label{cor:measurability}If ${\textup{(A2)}}$ holds, then $W_{n}^{i}$
must be measurable w.r.t. $\mathcal{F}_{n-1}$ for every $n\geq0$
and $i\in[N]$.
\end{cor}

\noindent\textit{Sequential importance sampling}: $\mathbb{A}_{N}=\{\mathrm{Id}\}$\vspace*{6pt}

Since in this case $\mathbb{A}_{N}$ consists of only a single element,
$\alpha$ is actually a deterministic sequence, {\textup{(A2)}} is trivially
satisfied and at line $(\star)$ of Algorithm \ref{alg:aSMC} we
have $\alpha_{n}=\mathrm{Id}$ fixed for all $n\geq0$. In this situation, Lemma~\ref{lem:W_n_representation} gives $W_{n}^{i}=\prod_{p=0}^{n-1}g_{p}(\zeta_{p}^{i})$
for $n\geq1$, so in turn
\[
\pi_{n}^{N}=\frac{\sum_{i}\delta_{\zeta_{n}^{i}}\prod_{p=0}^{n-1}g_{p}(\zeta_{p}^{i})}{\sum_{i}\prod_{p=0}^{n-1}g_{p}(\zeta
_{p}^{i})},\qquad Z_{n}^{N}=
\frac{1}{N}\sum_{i}\prod
_{p=0}^{n-1}g_{p}\bigl(
\zeta_{p}^{i}\bigr),\qquad n\geq1,
\]
and $\alpha$SMC reduces to Algorithm \ref{alg:SIS}.\vadjust{\goodbreak}
\begin{algorithm}
\begin{raggedright}
\quad For $n=0$,
\end{raggedright}

\begin{raggedright}
\qquad For $i=1,\ldots,N$,
\end{raggedright}

\begin{raggedright}
\qquad\quad Set $W_{0}^{i}=1$
\end{raggedright}

\begin{raggedright}
\qquad\quad Sample $\zeta_{0}^{i}\sim\mu_{0}$
\end{raggedright}

\begin{raggedright}
\quad For $n\geq1$,
\end{raggedright}

\begin{raggedright}
\qquad For $i=1,\ldots,N$,
\end{raggedright}

\begin{raggedright}
\qquad\quad Set $W_{n}^{i}=W_{n-1}^{i}g_{n-1}(\zeta
_{n-1}^{i})$
\end{raggedright}

\begin{raggedright}
\qquad\quad Sample $\zeta_{n}^{i}|\mathcal{F}_{n-1}
\sim  f(\zeta_{n-1}^{i},\cdot)$
\end{raggedright}
\caption{Sequential importance sampling}
\label{alg:SIS}\vadjust{\goodbreak}
\end{algorithm}

\noindent\textit{Bootstrap particle filter}: $\mathbb{A}_{N}=\{\mathbf
{1}_{1/N}\}$\vspace*{6pt}

In this case, $\alpha$ is again a deterministic sequence and \textup{(A2)}
is trivially satisfied. At line $(\star)$ we have $\alpha_{n}=\mathbf{1}_{1/N}$
fixed for all $n\geq0$. Lemma~\ref{lem:W_n_representation} gives,
for all $i_{n}\in[N]$,
%
\begin{equation}
W_{n}^{i_{n}}=\sum_{ (i_{0},\ldots,i_{n-1} )\in[N]^{n}}\prod
_{p=0}^{n-1}\frac{g_{p}(\zeta_{p}^{i_{p}})}{N}=\prod
_{p=0}^{n-1} \biggl(\frac{1}{N}\sum
_{i_{p}}g_{p}\bigl(\zeta_{p}^{i_{p}}
\bigr) \biggr), \qquad n\geq 1.\label{eq:bootstrap_W_n^i}
\end{equation}
Note that then $W_{n}^{i}=W_{n}^{j}$ for all $i,j$, so $NW_{n}^{i}=\sum_{j}W_{n}^{j}$
and we obtain, according to (\ref{eq:pi^N_andZ^N}),
%
\begin{equation}
\pi_{n}^{N}=\frac{1}{N}\sum
_{i}\delta_{\zeta_{n}^{i}}, \qquad Z_{n}^{N}=
\prod_{p=0}^{n-1} \biggl(\frac{1}{N}
\sum_{i_{p}}g_{p}\bigl(\zeta
_{p}^{i_{p}}\bigr) \biggr),\qquad  n\geq1,\label{eq:bootstrap_Z_n^N}
\end{equation}
and $\alpha$SMC algorithm reduces to Algorithm \ref{alg:boot_pf}. Since $W_{n}^{i}=W_{n}^{j}$
for all $i,j$, we write by convention the weight update steps only
for $W_{n}^{1}$.\vspace*{6pt}

\begin{algorithm}[b]
\begin{raggedright}
\quad For $n=0$,
\end{raggedright}

\begin{raggedright}
\qquad Set $W_{0}^{1}=1$
\end{raggedright}

\begin{raggedright}
\qquad For $i=1,\ldots,N$,
\end{raggedright}

\begin{raggedright}
\qquad\quad Sample $\zeta_{0}^{i}\sim\mu_{0}$
\end{raggedright}

\begin{raggedright}
\quad For $n\geq1$,
\end{raggedright}

\begin{raggedright}
\qquad Set $W_{n}^{1}=W_{n-1}^{1}\cdot\frac{1}{N}\sum_{i}g_{n-1}(\zeta_{n-1}^{i})$
\end{raggedright}

\begin{raggedright}
\qquad For $i=1,\ldots,N$,
\end{raggedright}

\begin{raggedright}
\qquad\quad Sample $\zeta_{n}^{i}|\mathcal{F}_{n-1}
\sim \frac{\sum_{j}g_{n-1}(\zeta_{n-1}^{j})f(\zeta_{n-1}^{j},\cdot
)}{\sum_{j}g_{n-1}(\zeta_{n-1}^{j})}$
\end{raggedright}
\caption{Bootstrap particle filter}
\label{alg:boot_pf}
\end{algorithm}

\noindent\textit{Adaptive resampling particle filter}: $\mathbb{A}_{N}=\{
\mathrm{Id},\mathbf{1}_{1/N}\}$\vspace*{6pt}

In this case, each $\alpha_{n}$ is allowed to take only the value
$\mathrm{Id}$ or $\mathbf{1}_{1/N}$, with the latter corresponding to resampling,
and vice-versa. The choice between $\mathrm{Id}$ and $\mathbf{1}_{1/N}$ is
made by comparing some functional of the particle system to a threshold
value. We consider the case of the popular ESS-based resampling rule
\cite{liu1995blind}, partly for simplicity, but also because monitoring
of the ESS is especially pertinent to the discussions which follow.
This ARPF arises as an instance of $\alpha$SMC if we take as line
$(\star)$ of Algorithm \ref{alg:aSMC} the rule:
%
\begin{equation}
\alpha_{n-1}:= %
\cases{\mathbf{1}_{1/N}, & \quad $\mbox{if } \displaystyle\frac{ (N^{-1}\sum_{i}W_{n-1}^{i}g_{n-1}(\zeta_{n-1}^{i}) )^{2}}{N^{-1}\sum_{i}
(W_{n-1}^{i}g_{n-1}(\zeta_{n-1}^{i}) )^{2}}<\tau$,\vspace*{2pt}
\cr
\mathrm{Id}, & \quad $\mbox{otherwise}$,}
\label{eq:alpha_ARPF}
\end{equation}
where $\tau\in(0,1]$ is a threshold value. Lemma~\ref{lem:ARPF_A2}
in the \hyperref[sec:Appendix]{Appendix} shows by an inductive argument that the adaptation
rule (\ref{eq:alpha_ARPF}) satisfies \textup{(A2)}. The ARPF is
traditionally expressed in terms of the random times at which resampling
occurs. For completeness, the \hyperref[sec:Appendix]{Appendix} contains derivations of expressions
for $\pi_{n}^{N}$ and $Z_{n}^{N}$ in terms of such times and similar
manipulations can be used to write out the form of $\alpha$SMC in
this case.

Looking back to the expression for $\mathcal{E}_{n}^{N}$ in (\ref
{eq:ESS_defn_front}),
we find:
%
\begin{eqnarray}
\alpha_{n-1}&=&\mathbf{1}_{1/N} \quad \Rightarrow \quad\mathcal
{E}_{n}^{N}=1,\label{eq:alpha_ARPF_1}
\\
\alpha_{n-1}&=&\mathrm{Id}\quad  \Rightarrow \quad\mathcal {E}_{n}^{N}=
\frac{ (N^{-1}\sum_{i}W_{n-1}^{i}g_{n-1}(\zeta
_{n-1}^{i}) )^{2}}{N^{-1}\sum_{i} (W_{n-1}^{i}g_{n-1}(\zeta
_{n-1}^{i}) )^{2}}.\label{eq:alpha_ARPF_2}
\end{eqnarray}
We then adopt the point of view that according to (\ref
{eq:alpha_ARPF})--(\ref{eq:alpha_ARPF_2}),
the ARPF \emph{enforces} the condition: $\inf_{n\geq0}\mathcal
{E}_{n}^{N}\geq\tau>0$,
or equivalently
\[
\inf_{n\geq0}N_{n}^{\mathrm{eff}}\geq N\tau>0,
\]
by construction. This seemingly trivial observation turns out to
be crucial when we address time-uniform convergence of the ARPF in
Section~\ref{sec:stability}, and the condition $\inf_{n\geq0}\mathcal
{E}_{n}^{N}>0$
will appear repeatedly in discussions which lead to the formulation
of new, provably stable algorithms in Section~\ref{sec:Discussion}.

To give some flavour of the kind of algorithms we have in mind, let
$ (B_{L}^{\ell} )_{\ell=1,\ldots, L}$ be a partition of the
set $[N]$ into $L$ clusters, and suppose the matrix $\alpha_{n-1}$
is defined by $\alpha_{n-1}^{ij}=1/|B_{L}^{\ell}|$ if both $i,j\in
B_{L}^{\ell}$
and $0$ otherwise. Then for any $\ell=1,\ldots,L$ and any $i\in
B_{L}^{\ell}$
the new weight $W_{n}^{i}$ and the distribution from which the new
particle $\zeta_{n}^{i}$ is sampled, say $m_{n}^{i}$ (both of which
depend on $\ell$ only), are given by
%
\begin{equation}
W_{n}^{i}=\frac{1}{|B_{L}^{\ell}|}\sum
_{j\in B_{L}^{\ell
}}W_{n-1}^{j}g_{n-1}\bigl(
\zeta_{n-1}^{j}\bigr),\qquad  m_{n}^{i}=
\frac{\sum_{j\in
B_{L}^{\ell}}W_{n-1}^{j}g_{n-1}(\zeta_{n-1}^{j})f(\zeta_{n-1}^{j},\cdot
)}{\sum_{j\in B_{L}^{\ell}}W_{n-1}^{j}g_{n-1}(\zeta_{n-1}^{j})}.\label
{eq:partition_intro}
\end{equation}
Furthermore, in this situation we have
\[
\mathcal{E}_{n}^{N}=\frac{ (N^{-1}\sum_{i}W_{n-1}^{i}g_{n-1}(\zeta
_{n-1}^{i}) )^{2}}{\sum_{\ell=1}^{L}|B_{L}^{\ell}|/N
(|B_{L}^{\ell}|^{-1}\sum_{j\in B_{L}^{\ell}}W_{n-1}^{j}g_{n-1}(\zeta
_{n-1}^{j}) )^{2}}.
\]
It is then clear that a parallel implementation could be possible,
say on $L$ processors, one devoted to each cluster and it remains
to design an efficient partition of $ (B_{L}^{\ell} )_{\ell
=1,\ldots, L}$
of the set $[N]$.\vspace*{6pt}

\noindent\textit{Comments on other algorithms}\vspace*{6pt}

In the engineering literature, a variety of algorithmic procedures
involving distributed computing have been suggested \cite{bolic2005resampling},
including partitioning ideas like (\ref{eq:partition_intro}). ``Local''
particle approximations of Rao--Blackwellized filters have been devised
in \cite{chen2011decentralized} and \cite{johansen2012exact}.
Verg\'e \textit{et al.} \cite{Verge_island_particle} have recently suggested an ``island''
particle algorithm, designed for parallel implementation, in which
there are two levels of resampling and the total population size $N=N_{1}N_{2}$
is defined in terms of the number of particles per island, $N_{1}$,
and the number of islands, $N_{2}$. Interaction at both levels occurs
by resampling, at the island level this means entire blocks of particles
are replicated and/or discarded. They investigate the trade-off between
$N_{1}$ and $N_{2}$ and provide asymptotic results which validate
their algorithms. In the present work, we provide some asymptotic
results in Section~\ref{sec:Martingale-approximations-and} but it
is really the non-asymptotic results in Section~\ref{sec:stability}
which lead us to suggest specific novel instances of $\alpha$SMC
in Section~\ref{sec:Discussion}. Moreover, in general $\alpha$SMC
is distinct from all these algorithms and, other than in some uninteresting
special cases, none of them coincide with the adaptive procedures
we suggest in Section~\ref{sub:Algorithms-with-adaptive}.

\section{Convergence}
\label{sec:Martingale-approximations-and}

In this section, our main objective is to investigate, for general
$\alpha$SMC (Algorithm \ref{alg:aSMC}), conditions for convergence
%
\begin{equation}
Z_{n}^{N}-Z_{n}\rightarrow0 \quad\mbox{and}\quad
\pi_{n}^{N}(\varphi)-\pi _{n}(\varphi)
\rightarrow0,\label{eq:as_convergence_intro}
\end{equation}
at least in probability, as $N\rightarrow\infty$.

In the case of SIS, that is, $\mathbb{A}_{N}=\{\mathrm{Id}\}$, it is easy to
establish (\ref{eq:as_convergence_intro}), since the processes $ \{
\zeta_{n}^{i};n\geq0 \} _{i\in[N]}$
are independent Markov chains, of identical law. On the other hand,
for the bootstrap filter, that is, $\mathbb{A}_{N}=\{\mathbf{1}_{1/N}\}$,
the convergence $\pi_{n}^{N}(\varphi)-\pi_{n}(\varphi)\rightarrow0$,
can be proved under very mild conditions, by decomposing $\pi
_{n}^{N}(\varphi)-\pi_{n}(\varphi)$
in terms of ``local'' sampling errors, see amongst others \cite
{smc:theory:Dm04,smc:the:DM08}
for this type of approach. For instance, for $A\in\mathcal{X}$ we
may write
%
\begin{eqnarray}
\pi_{1}^{N}(A)-\pi_{1}(A) & = & \frac{1}{N}
\sum_{i}\delta_{\zeta
_{1}^{i}}(A)-
\frac{\sum_{i}g_{0}(\zeta_{0}^{i})f(\zeta_{0}^{i},A)}{\sum_{i}g_{0}(\zeta_{0}^{i})}\label{eq:intro_boot_decomp1}
\\
&&{} +  \frac{\sum_{i}g_{0}(\zeta_{0}^{i})f(\zeta_{0}^{i},A)}{\sum_{i}g_{0}(\zeta_{0}^{i})}-\pi_{1}(A).\label{eq:intro_boot_decomp2}
\end{eqnarray}
Heuristically, the term on the r.h.s. of (\ref{eq:intro_boot_decomp1})
converges to zero because given $\mathcal{F}_{0}$, the samples $ \{
\zeta_{1}^{i} \} _{i\in[N]}$
are conditionally i.i.d. according $\frac{\sum_{i}g_{0}(\zeta
_{0}^{i})f(\zeta_{0}^{i},\cdot)}{\sum_{i}g_{0}(\zeta_{0}^{i})}$,
and the term in (\ref{eq:intro_boot_decomp2}) converges to zero because
the samples $ \{ \zeta_{0}^{i} \} _{i\in[N]}$ are i.i.d. according
to $\mu_{0}$. A similar argument ensures that $\pi_{n}^{N}(\varphi)-\pi
_{n}(\varphi)\rightarrow0$,
for any $n\geq0$ and therefore by the continuous mapping theorem
$Z_{n}^{N}-Z_{n}\rightarrow0$, since
\[
Z_{n}=\prod_{p=0}^{n-1}
\pi_{p}(g_{p})\quad \mbox{and}\quad Z_{n}^{N}=
\prod_{p=0}^{n-1}\pi_{p}^{N}(g_{p}).
\]
In the case of $\alpha$SMC, $ \{ \zeta_{n}^{i} \} _{i\in[N]}$
are conditionally independent given $\mathcal{F}_{n-1}$, but we do
not necessarily have either the unconditional independence structure
of SIS, or the conditionally i.i.d. structure of the BPF to work
with.

Douc and Moulines \cite{smc:the:DM08} have established a CLT for the ARPF using an
inductive approach w.r.t. deterministic time periods. Arnaud and Le Gland \cite{arnaud2009smc}
have obtained a CLT for the ARPF based on an alternative multiplicative
functional representation of the algorithm. Convergence of the ARPF
was studied in \cite{del2012adaptive} by coupling the adaptive algorithm
to a reference particle system, for which resampling occurs at deterministic
times. One of the benefits of their approach is that existing asymptotic
results for non-adaptive algorithms, such as central limit theorems
(CLT), can then be transferred to the adaptive algorithm with little
further work. Their analysis involves a technical assumption \cite{del2012adaptive},
Section~5.2,
to deal with the situation where the threshold parameters coincide
with the adaptive criteria. Our analysis of $\alpha$SMC does not
rest on any such technical assumption, and in some ways is more direct,
but we do not obtain concentration estimates or a CLT. Some more detailed
remarks on this matter are given after the statement of Theorem~\ref
{thmm:convergence}.

Crisan and Obanubi \cite{crisan2012particle} studied convergence and obtained a CLT
for an adaptive resampling particle filter in continuous time under
conditions which they verify for the case of ESS-triggered resampling,
without needing the type of technical assumption of \cite{del2012adaptive}.
Their study focuses, in part, on the random times at which resampling
occurs and dealing with the subtleties of the convergence in continuous
time. Our asymptotic $N\rightarrow\infty$ analysis is in some ways
less refined, but in comparison to this and the other existing works,
we analyze a more general algorithm, and it is this generality which
allows us to suggest new adaptive algorithms in Section~\ref{sec:Discussion},
informed by the time-uniform non-asymptotic error bounds in our
Theorem~\ref{thmm:L_R_mix}.

To proceed, we need some further notation involving $\alpha$. Let
us define the matrices: $\alpha_{n,n}:=\mathrm{Id}$ for $n\geq0$, and recursively
%
\begin{equation}
\alpha_{p,n}^{ij}:=\sum_{k}
\alpha_{p+1,n}^{ik}\alpha_{p}^{kj},\qquad (i,j)
\in[N]^{2}, 0\leq p<n,\label{eq:a_pn_defn}
\end{equation}
and the vectors:
%
\begin{equation}
\beta_{n,n}^{i}:=N^{-1},\qquad n\geq0, i
\in[N],\label{eq:beta_n_n_defn}
\end{equation}
and recursively
%
\begin{equation}
\beta_{p,n}^{i}:=\sum_{j}
\beta_{p+1,n}^{j}\alpha_{p}^{ji},\qquad i\in [N],
0\leq p<n.\label{eq:beta_defn}
\end{equation}
Note that since each $\alpha_{n}$ is a random Markov transition matrix,
so is each $\alpha_{p,n}$, and each $ \{ \beta_{p,n}^{i} \}
_{i\in[N]}$
defines a random probability distribution on $[N]$. Moreover, from
these definitions we immediately have the identity
%
\begin{equation}
\beta_{p,n}^{i}=N^{-1}\sum
_{j}\alpha_{p,n}^{ji},\qquad i\in[N], 0\leq p
\leq n.\label{eq:beta_in_terms_of_alpha}
\end{equation}

\begin{assumption*}
\textup{(B)} -- for all $0\leq p\leq n$ and $i\in[N]$, $\beta
_{p,n}^{i}$
is measurable w.r.t. the trivial $\sigma$-algebra $\mathcal{F}_{-1}$.

\textup{(B$^{+}$)} -- assumption \textup{(B)} holds and,
for all $0\leq p\leq n$, $\lim_{N\rightarrow\infty}\max_{i\in[N]}\beta
_{p,n}^{i}=0$.

\textup{(B$^{++}$)} -- every member of $\mathbb{A}_{N}$
admits the uniform distribution on $[N]$ as an invariant distribution
\end{assumption*}
We note the following:
\begin{itemize}
\item Intuitively, (B) ensures that even though $\alpha$ is
a sequence of random Markov transition matrices, the elements of the
probability vector $ \{ \beta_{p,n}^{i} \} _{i\in[N]}$ are
all constants. (B) holds, trivially, when every element
of every $\alpha_{n}$ is measurable w.r.t. $\mathcal{F}_{-1}$, that
is, the sequence $\alpha$ is completely pre-determined. This is
true, for example, when the set $\mathbb{A}_{N}$ consists of only
a single element, as is the case for SIS and the BPF.
\item The $\lim_{N\rightarrow\infty}\max_{i\in[N]}\beta_{p,n}^{i}=0$ part
of (B$^{+}$) is an asymptotic negligibility condition. In
Section~\ref{sub:Ensuring-convergence}, we describe what can go wrong
when this assumption does not hold.
\item (B$^{++}$) does not require the members of $\mathbb{A}_{N}$
to be irreducible, for example, it is satisfied with $\mathbb{A}_{N}=\{
\mathrm{Id}\}$.
\item (B$^{++}$) $\Rightarrow$ (B$^{+}$). To
see this, observe that when (B$^{++}$) holds, every random
matrix $\alpha_{p,n}$, defined in (\ref{eq:a_pn_defn}), also admits
the uniform distribution on $[N]$ as invariant, then using (\ref
{eq:beta_in_terms_of_alpha})
we have $\beta_{p,n}^{i}=N^{-1}\sum_{j}\alpha_{p,n}^{ji}=N^{-1}$
for all $i\in[N]$. The reverse implication is clearly not true in
general.
\item (B$^{++}$) holds when every member of $\mathbb{A}_{N}$
is doubly-stochastic, because such matrices always leave the uniform
distribution invariant. (B$^{++}$) therefore holds
for the ARPF, which has $\mathbb{A}_{N}=\{\mathrm{Id},\mathbf{1}_{1/N}\}$.
\end{itemize}
To get some feel for why (B$^{++}$) is a natural condition
for convergence, note that plugging the particle approximation $\pi
_{n-1}\approx\pi_{n-1}^{N}=\sum_{i}W_{n-1}^{i}\delta_{\zeta
_{n-1}^{i}}/\sum_{i}W_{n-1}^{i}$into
equation (\ref{eq:filtering_recursion}) for the predictor yields
a finite mixture approximation of $\pi_{n}$
\[
\pi_{n}\approx\frac{\sum_{i}W_{n-1}^{i}g_{n-1}(\zeta_{n-1}^{i})f(\zeta
_{n-1},\cdot)}{\sum_{i}W_{n-1}^{i}g_{n-1}(\zeta_{n-1}^{i})}.
\]
Under condition (B$^{++}$) the stochastic matrix $\alpha
_{n-1}\in\mathbb{A}_{N}$
is doubly stochastic, hence
\[
\frac{\sum_{j}W_{n-1}^{j}g_{n-1}^{j}(\zeta_{n-1}^{j})f(\zeta
_{n-1}^{j},\cdot)}{\sum_{j}W_{n-1}^{j}g_{n-1}^{j}(\zeta
_{n-1}^{j})}=\frac{\sum_{i}\sum_{j}\alpha
_{n-1}^{ij}W_{n-1}^{j}g_{n-1}^{j}(\zeta_{n-1}^{j})f(\zeta
_{n-1}^{j},\cdot)}{\sum_{i}\sum_{j}\alpha
_{n-1}^{ij}W_{n-1}^{j}g_{n-1}^{j}(\zeta_{n-1}^{j})}=\frac{\sum_{i}W_{n}^{i}m_{n}^{i}}{\sum_{i}W_{n}^{i}},
\]
where $m_{n}^{i}$ is the distribution from which the new particle
$\zeta_{n}^{i}$ is sampled, as in (\ref{eq:partition_intro}), justifying
the particle approximation
\[
\pi_{n}\approx\pi_{n}^{N}=\frac{\sum_{i}W_{n}^{i}\delta_{\zeta
_{n}^{i}}}{\sum_{i}W_{n}^{i}}.
\]

The main result of this section is:

\begin{thmm}
\label{thmm:convergence} $ $Assume \textup{(A2)}. For any $n\geq0$,
$\varphi\in\mathcal{L}$ and $r\geq1$,

\begin{longlist}[(1)]
\item[(1)] if \textup{(B)} holds, then $\mathbb{E}[Z_{n}^{N}]=Z_{n}$ for
any $N\geq1$,

\item[(2)] if \textup{(B$^{+}$)} holds, then
%
\begin{eqnarray}
 \lim_{N\rightarrow\infty}\mathbb{E} \bigl[\bigl\llvert
Z_{n}^{N}-Z_{n}\bigr\rrvert ^{r}
\bigr]&=&0,\label
{eq:convergence_L_r_statement_gam_weak}
\\
\lim_{N\rightarrow\infty}\mathbb{E} \bigl[\bigl\llvert
\pi_{n}^{N}(\varphi )-\pi_{n}(\varphi)\bigr\rrvert
^{r} \bigr]&=&0,\label
{eq:convergence_L_r_statement_pi_weak}
\end{eqnarray}
and therefore $Z_{n}^{N}\rightarrow Z_{n}$ and $\pi
_{n}^{N}(\varphi)\rightarrow\pi_{n}(\varphi)$
in probability as $N\rightarrow\infty$,

\item[(3)] if \textup{(B$^{++}$)} holds, then
%
\begin{eqnarray}
\sup_{N\geq1}\sqrt{N}\mathbb{E} \bigl[\bigl\llvert
Z_{n}^{N}-Z_{n}\bigr\rrvert ^{r}
\bigr]^{1/r}&<&+\infty,\label{eq:convergence_L_r_statement_gamm}
\\
\sup_{N\geq1}\sqrt{N}\mathbb{E} \bigl[\bigl\llvert
\pi_{n}^{N}(\varphi)-\pi _{n}(\varphi)\bigr\rrvert
^{r} \bigr]^{1/r}&<&+\infty,\label
{eq:convergence_L_r_statement_pi}
\end{eqnarray}
and therefore $Z_{n}^{N}\rightarrow Z_{n}$ and $\pi
_{n}^{N}(\varphi)\rightarrow\pi_{n}(\varphi)$
almost surely, as $N\rightarrow\infty$.
\end{longlist}
\end{thmm}

\begin{rem}
The lack-of-bias property $\mathbb{E}[Z_{n}^{N}]=Z_{n}$ is desirable
since it could be used to validate the use of $\alpha$SMC within
composite SMC/MCMC algorithms such as those of~\cite{andrieu2010particle}.
\end{rem}

\begin{rem}
\label{rem_alpha_conv}Theorem~\ref{thmm:convergence} holds without
any sort of requirement that the entries of each $\alpha_{n}$ converge
as $N\rightarrow\infty$. For example, (B$^{++}$) holds
if for $N$ odd we choose $\mathbb{A}_{N}=\{\mathrm{Id}\}$ and for $N$ even
we choose $\mathbb{A}_{N}=\{\mathbf{1}_{1/N}\}$. As a reflection
of this, and as is apparent upon inspection of the proof, without
further assumption we cannot in general replace $\sup_{N\geq1}$ in
(\ref{eq:convergence_L_r_statement_gamm})--(\ref
{eq:convergence_L_r_statement_pi})
with $\lim_{N\rightarrow\infty}$, because such limits may not exist.
\end{rem}

The following notation is used throughout the remainder of the paper.
Introduce the non-negative kernels
%
\begin{equation}
Q_{n}:\mathsf{X}\times\mathcal{X}\rightarrow\mathbb{R}_{+},\qquad
Q_{n}\bigl(x,\mathrm{d}x^{\prime}\bigr):=g_{n-1}(x)f
\bigl(x,\mathrm{d}x^{\prime}\bigr),\qquad n\geq1,\label
{eq:Q_GM_defn}
\end{equation}
the corresponding operators on functions and measures:
%
\begin{eqnarray}
Q_{n}(\varphi) (x) &:= & \int_{\mathsf{X}}Q_{n}
\bigl(x,\mathrm{d}x^{\prime}\bigr)\varphi \bigl(x^{\prime}\bigr), \qquad\varphi\in
\mathcal{L},\label{eq:Q_op_defn}
\\
\mu Q_{n}(\cdot) & := & \int_{\mathsf{X}}
\mu(\mathrm{d}x)Q_{n}(x,\cdot),\qquad \mu \in\mathcal{M},\label{eq:Q_op_defn-1}
\end{eqnarray}
and for $n\geq1$ and $0\leq p<n$,
%
\begin{equation}
Q_{p,p}:=\mathrm{Id},\qquad Q_{p,n}:=Q_{p+1}\cdots
Q_{n}.\label{eq:Q_semigroup}
\end{equation}
We shall also consider the following scaled versions of these operators:
%
\begin{equation}
\overline{Q}_{n}:=\frac{Q_{n}}{\pi_{n-1}(g_{n-1})},\qquad \overline
 {Q}_{p,p}:=\mathrm{Id},\qquad
\overline{Q}_{p,n}:=\overline{Q}_{p+1}\cdots
\overline{Q}_{n}.\label{eq:Q_bar_defn}
\end{equation}

Then define the non-negative measures
\[
\gamma_{n}:=\mu_{0}Q_{0,n}(\cdot),\qquad n\geq0,
\]
under (A1) we are assured that $\gamma_{n}(1)>0$.
Due to the conditional independence structure of the HMM, it can easily
be checked that
\[
\pi_{n}=\frac{\gamma_{n}}{\gamma_{n}(1)},\qquad Z_{n}=\gamma
_{n}(1),\qquad n\geq0
\]
and
\[
\overline{Q}_{p,n}=\frac{Q_{p,n}}{\pi_{p}Q_{p,n}(1)}.
\]

For $i\in[N]$ and $0\leq p\leq n$, introduce the random measures
%
\begin{equation}
\Gamma_{p,n}^{N}:=\sum_{i}
\beta_{p,n}^{i}W_{p}^{i}
\delta_{\zeta
_{p}^{i}}, \qquad \overline{\Gamma}_{p,n}^{N}:=
\frac{\Gamma
_{p,n}^{N}}{\gamma_{p}(1)},\label{eq:Gamma_defn}
\end{equation}
where $W_{p}^{i}$ is as in (\ref{eq:W_n_defn}). For simplicity of
notation, we shall write $\Gamma_{n}^{N}:=\Gamma_{n,n}^{N}, \overline
{\Gamma}_{n}^{N}:=\overline{\Gamma}_{n,n}^{N}$.
If we define
%
\begin{equation}
\overline{W}_{n}^{i}:=\frac{W_{n}^{i}}{\gamma_{n}(1)},\qquad n\geq
0,\label{eq:W_bar_defn}
\end{equation}
then we have from (\ref{eq:Gamma_defn}),
\[
\overline{\Gamma}_{p,n}^{N}=\sum
_{i}\beta_{p,n}^{i}\overline
{W}_{p}^{i}\delta_{\zeta_{p}^{i}}.
\]

Finally, we observe from (\ref{eq:beta_n_n_defn}) that
\[
\Gamma_{n}^{N}=\sum_{i}
\beta_{n,n}^{i}W_{n}^{i}
\delta_{\zeta
_{n}^{i}}=N^{-1}\sum_{i}W_{n}^{i}
\delta_{\zeta_{n}^{i}}.
\]

\subsection*{Error decomposition}
Throughout this section, let $\varphi\in\mathcal{L}$, $n\geq0$ and
$N\geq1$ be arbitrarily chosen, but then fixed. Define, for $1\leq
p\leq n$
and $i\in[N]$,
\[
\Delta_{p,n}^{i}:=\overline{Q}_{p,n}(\varphi)
\bigl(\zeta_{p}^{i}\bigr)-\frac{\sum_{j}\alpha_{p-1}^{ij}W_{p-1}^{j}\overline{Q}_{p-1,n}(\varphi)(\zeta
_{p-1}^{j})}{\sum_{j}\alpha_{p-1}^{ij}W_{p-1}^{j}\overline
{Q}_{p}(1)(\zeta_{p-1}^{j})},
\]
and $\Delta_{0,n}^{i}:=\overline{Q}_{0,n}(\varphi)(\zeta_{0}^{i})-\mu
_{0}\overline{Q}_{0,n}(\varphi)$,
so that $\mathbb{E} [\Delta_{p,n}^{i}|\mathcal
{F}_{p-1} ]=0$
for any $i\in[N]$ and $0\leq p\leq n$. Then for $0\leq p\leq n$
and $i\in[N]$ set $k:=pN+i$, and
\[
\xi_{k}^{N}:=\sqrt{N}\beta_{p,n}^{i}
\overline{W}_{p}^{i}\Delta_{p,n}^{i},
\]
so as to define a sequence $ \{ \xi_{k}^{N};k=1,\ldots,(n+1)N
\} $.
For $k=1,\ldots,(n+1)N$, let $\mathcal{F}^{(k)}$ be the $\sigma$-algebra
generated by $ \{ \zeta_{p}^{i};  pN+i\leq k,  i\in[N],0\leq
p\leq n \} $.
Set $\mathcal{F}^{(-1)}:=\{\mathbb{X},\varnothing\}$.

The following proposition is the main result underlying Theorem~\ref
{thmm:convergence}.
The proof is given in the \hyperref[sec:Appendix]{Appendix}.

\begin{prop}
\label{prop:martingale} Assume \textup{(A2)} and \textup{(B)}.
We have the decomposition
%
\begin{equation}
\sqrt{N} \bigl[\overline{\Gamma}_{n}^{N}(\varphi)-
\pi_{n}(\varphi) \bigr]=\sum_{k=1}^{(n+1)N}
\xi_{k}^{N},\label{eq:Gamma_telescope}
\end{equation}
where for $k=1,\ldots,(n+1)N$, the increment $\xi_{k}^{N}$ is measurable
w.r.t. $\mathcal{F}^{(k)}$ and satisfies
%
\begin{equation}
\mathbb{E} \bigl[\xi_{k}^{N}|\mathcal{F}^{(k-1)}
\bigr]=\mathbb{E} \bigl[\xi_{k}^{N}|\mathcal{F}_{p-1}
\bigr]=0\qquad \mbox{with } p:= \bigl\lfloor(k-1)/N \bigr\rfloor.\label
{eq:xi_cond_exp}
\end{equation}
For each $r\geq1$ there exists a universal constant $B(r)$ such
that
%
\begin{eqnarray}\label{eq:martingale_burkholder_bound}
& & \mathbb{E} \bigl[\bigl\llvert \overline{\Gamma}_{n}^{N}(
\varphi)-\pi _{n}(\varphi)\bigr\rrvert ^{r}
\bigr]^{1/r}
\nonumber
\\[-8pt]
\\[-8pt]
\nonumber
& &\quad \leq B(r)^{1/r}\sum_{p=0}^{n}
\mathrm{osc} \bigl(\overline {Q}_{p,n}(\varphi) \bigr)\mathbb{E} \biggl[
\biggl\llvert \sum_{i} \bigl(\beta
_{p,n}^{i}\overline{W}_{p}^{i}
\bigr)^{2}\biggr\rrvert ^{r/2} \biggr]^{1/r}.
\end{eqnarray}
\end{prop}

The proof of Theorem~\ref{thmm:convergence}, which is mostly technical,
is given in the \hyperref[sec:Appendix]{Appendix}. Here we briefly discuss our assumptions
and sketch some of the main arguments. Part (1) of Theorem~\ref{thmm:convergence}
follows immediately from (\ref{eq:Gamma_telescope}) and (\ref{eq:xi_cond_exp})
applied with $\varphi=1$. In turn, the martingale structure of
(\ref{eq:Gamma_telescope}) and (\ref{eq:xi_cond_exp}) is underpinned
by the measurability conditions \textup{(A2)} and (B).
The proofs of parts (2) and (3) of Theorem~\ref{thmm:convergence},
involve applying Proposition~\ref{prop:martingale} in conjunction
with the identities
%
\begin{eqnarray}\label
{eq:convergence_sketch_id}
Z_{n}^{N}-Z_{n} & = & \Gamma_{n}^{N}(1)-
\gamma_{n}(1),
\nonumber
\\[-8pt]
\\[-8pt]
\nonumber
\pi_{n}^{N}(\varphi)-\pi_{n}(\varphi) & = &
\frac{\Gamma_{n}^{N}(\varphi
)}{\Gamma_{n}^{N}(1)}-\frac{\gamma_{n}(\varphi)}{\gamma_{n}(1)}.
\end{eqnarray}
In order to prove that these errors convergence to zero in probability,
we show that the quadratic variation term in (\ref
{eq:martingale_burkholder_bound})
converges to zero. In general, we cannot hope for the latter convergence
without some sort of negligibility hypothesis on the product terms
$ \{ \mathrm{osc} (\overline{Q}_{p,n}(\varphi) )\beta
_{p,n}^{i}\overline{W}_{p}^{i};i\in[N] \} $.
Assumption (A1) allows us to crudely upper-bound
$\mathrm{osc} (\overline{Q}_{p,n}(\varphi) )$ and $\overline
{W}_{p}^{i}$;
the measurability condition (B) allows us to dispose of
the expectation in (\ref{eq:martingale_burkholder_bound}); then via
Markov's inequality and the classical equivalence:
\[
\lim_{N\rightarrow\infty}\max_{i\in[N]}\beta_{p,n}^{i}=0
\quad\Leftrightarrow\quad \lim_{N\rightarrow\infty}\sum_{i}
\bigl(\beta _{p,n}^{i} \bigr)^{2}=0,
\]
which holds since $ (\max_{i\in[N]}\beta_{p,n}^{i} )^{2}\leq
\sum_{i} (\beta_{p,n}^{i} )^{2}\leq\max_{i\in[N]}\beta_{p,n}^{i}$,
the negligibility part of (B$^{+}$) guarantees
that $\llvert \Gamma_{n}^{N}(\varphi)-\gamma_{n}(\varphi)\rrvert $
converges to zero in probability. The stronger condition (B$^{++}$)
buys us the $\sqrt{N}$ scaling displayed in part (3). In Section~\ref{sub:Ensuring-convergence},
we discuss what can go wrong when (B$^{+}$) does not hold.

\section{Stability}
\label{sec:stability}

In this section, we study the stability of approximation errors under
the following regularity condition.
\begin{assumption*}
\textup{(C)} There exists $ (\delta,\epsilon )\in[1,\infty)^{2}$
such that
\[
\sup_{n\geq0}\sup_{x,y}\frac{g_{n}(x)}{g_{n}(y)}\leq
\delta,\qquad f(x,\cdot)\leq\epsilon f(y,\cdot),\qquad (x,y)\in\mathsf{X}^{2}.
\]
\end{assumption*}
(C) is a standard hypothesis in studies of non-asymptotic
stability properties of SMC algorithms. Similar conditions have been
adopted in \cite{smc:theory:Dm04}, Chapter~7, and \cite{smc:the:LGO04},
amongst others. (C)~guarantees that $Q_{p,n}$, and related
objects, obey a variety of regularity conditions. In particular, we
immediately obtain
%
\begin{equation}
\sup_{p,n}\sup_{x}\overline{Q}_{p,n}(1)
(x)\leq\sup_{p,n}\sup_{x,y}
\frac
{Q_{p,n}(1)(x)}{Q_{p,n}(1)(y)}\leq\delta\epsilon<+\infty.\label
{eq:Q_p,n_bounded}
\end{equation}
Furthermore, if we introduce the following operators on probability
measures:
%
\begin{eqnarray}
\Phi_{n}\dvtx \mu\in\mathcal{P}&\mapsto&\frac{\mu Q_{n}}{\mu(g_{n-1})}\in
\mathcal{P},\qquad n\geq1,\label{eq:Phi_defn}
\\
\Phi_{p,n}&:=&\Phi_{n}\circ\cdots\circ\Phi_{p+1},\qquad  0
\leq p<n.\label{eq:Phi_defn_semigroup}
\end{eqnarray}
It is well-known that under (C), $\Phi_{p,n}$ is uniformly
exponentially stable, in the sense of the somewhat crude estimate
in the following lemma.

\begin{lem}
\label{lem:Phi_stable}Assume \textup{(C)}. Then there exists a
finite constant $C$ and $\rho\in[0,1)$ such that
\[
\sup_{\mu,\mu^{\prime}\in\mathcal{P}}\bigl\llVert \Phi_{p,n}(\mu)-\Phi
_{p,n}\bigl(\mu^{\prime}\bigr)\bigr\rrVert _{\mathrm{tv}}\leq C
\rho^{n-p}.
\]
\end{lem}

For a proof see, for example, \cite{smc:theory:Dm04}, Proposition~4.3.6.
It follows from (\ref{eq:filtering_recursion}), (\ref{eq:Phi_defn})
and (\ref{eq:Phi_defn_semigroup}) that
\[
\pi_{n+1}=\Phi_{n+1}(\pi_{n})=
\Phi_{p,n+1}(\pi_{p})=\Phi_{0,n+1}(\mu
_{0}),\qquad 0\le p\leq n,
\]
so Lemma~\ref{lem:Phi_stable} can be used to describe the forgetting
of the initial distribution of the non-linear filter. Properties similar
to (\ref{eq:Q_p,n_bounded}) and the exponential stability in Lemma~\ref{lem:Phi_stable} can be obtained under conditions weaker and
more realistic than (C), see, for example, \cite{whiteley2013}
but the developments involved are substantially more technical, lengthy
and complicated to present. Our aim is to expedite the presentation
of stability properties of $\alpha$SMC, and (C) allows
this to be achieved while retaining some of the essence of more realistic
hypotheses on $g_{n}$ and $f$.

The main result of this section is the following theorem, whose proof
we briefly postpone.

\begin{thmm}
\label{thmm:L_R_mix}Assume \textup{(A2)}, \textup{(B$^{++}$)} and
\textup{(C)}. Then there exist finite constants, $c_{1}$ and for
each $r\geq1$, $c_{2}(r)$, such that for any $\tau\in(0,1]$, $N\geq1$,
and $\varphi\in\mathcal{L}$,
\begin{equation}
\inf_{n\geq0}\mathcal{E}_{n}^{N}\geq\tau \quad\Rightarrow\quad
\cases{
 \displaystyle \sup_{n\geq1} \mathbb{E}\biggl[\biggl(\frac{Z_{n}^{N}}{Z_{n}}
\biggr)^{2}\biggr]^{1/n} \leq 1+\frac{c_{1}}{N\tau}\quad\mathrm{and}\vspace*{2pt}\cr
  \sup_{n\geq0} \mathbb{E}\bigl[\bigl|\pi_{n}^{N}(\varphi)-\pi
_{n}(\varphi)\bigr|^{r}\bigr]^{1/r} \leq \Vert\varphi
\Vert\displaystyle\frac{c_{2}(r)}{\sqrt{N\tau}}.}
\label{eq:them_Stability_Statement}
\end{equation}
\end{thmm}

\begin{rem}
\label{rem:linear_variance}It follows quite immediately from the
first inequality of (\ref{eq:them_Stability_Statement}) that
\begin{eqnarray*}
\left.\begin{array} {r} \displaystyle\inf_{n\geq0}\mathcal{E}_{n}^{N}
\geq\tau
\\
\mbox{and}
\\
N\tau\geq nc_{1} \end{array} %
\right\} \quad \Rightarrow \quad\mathbb{E}
\biggl[ \biggl(\frac
{Z_{n}^{N}}{Z_{n}}-1 \biggr)^{2} \biggr]\leq
\frac{2nc_{1}}{N\tau},
\end{eqnarray*}
see Lemma~\ref{lem:var_bound} in the \hyperref[sec:Appendix]{Appendix}.
\end{rem}

\begin{rem}
It follows immediately from the second inequality in (\ref
{eq:them_Stability_Statement})
that when $\inf_{n\geq0}\mathcal{E}_{n}^{N}\geq\tau$ for all $N\geq1$,
the prediction filter errors are time-uniformly convergent in the
sense
\[
\lim_{N\rightarrow\infty}\sup_{n\geq0}\mathbb{E} \bigl[\bigl
\llvert \pi _{n}^{N}(\varphi)-\pi_{n}(\varphi)
\bigr\rrvert ^{r} \bigr]^{1/r}=0.
\]
\end{rem}

\begin{rem}
Further to the discussion of Section~\ref{sub:Instances-of-SMC},
in the case of the BPF we have $\mathcal{E}_{n}^{N}=1$ and hence
$\inf_{n\geq0}\mathcal{E}_{n}^{N}\geq\tau$ always, and for the ARPF
we also have $\inf_{n\geq0}\mathcal{E}_{n}^{N}\geq\tau$ always, by
virtue of the ESS rule for selection of $\alpha_{n}$. In Section~\ref{sec:Discussion},
we shall introduce new algorithms designed to guarantee $\inf_{n\geq
0}\mathcal{E}_{n}^{N}\geq\tau$.
\end{rem}

\begin{rem}
It is possible to deduce estimates for the constants $c_{1}$ and
$c_{2}(r)$ using the statements and proofs of Propositions \ref
{prop:norm_const_bound}
and \ref{prop:L_p_bound_mixing}, which are the main ingredients
in the proof of Theorem~\ref{thmm:L_R_mix}. We omit such expressions
only for simplicity of presentation.
\end{rem}

The proofs of Propositions \ref{prop:norm_const_bound} and \ref
{prop:L_p_bound_mixing}
are given in the \hyperref[sec:Appendix]{Appendix}.

\begin{prop}
\label{prop:norm_const_bound}Assume \textup{(A2)}, \textup{(B$^{++}$)} and
\textup{(C)}. If
for some sequence of constants $ \{ \tau_{n};n\geq0 \} \in
(0,1]^{\mathbb{N}}$
and $N\geq1$,
\[
\mathcal{E}_{n}^{N}\geq\tau_{n},
\]
then for any $n\geq1$,
\[
\mathbb{E} \biggl[ \biggl(\frac{Z_{n}^{N}}{Z_{n}}-1 \biggr)^{2} \biggr]
\leq \sum_{p=0}^{n-1}\frac{\mathrm{osc} (\overline{Q}_{p,n}(1)
)^{2}}{N\tau_{p}}
\biggl(\mathbb{E} \biggl[ \biggl(\frac
{Z_{p}^{N}}{Z_{p}}-1 \biggr)^{2}
\biggr]+1 \biggr).
\]
\end{prop}
%

\begin{prop}
\label{prop:L_p_bound_mixing}Consider the constants and Markov kernels:
\[
\delta_{p,n}:=\sup_{x,y}\frac{Q_{p,n}(1)(x)}{Q_{p,n}(1)(y)},\qquad
P_{p,n}(x,A):=\frac{Q_{p,n}(\mathbb{I}_{A})(x)}{Q_{p,n}(1)(x)},\qquad x\in\mathsf{X},A\in\mathcal{X}, 0\leq
p\leq n.
\]
Assume \textup{(A2)}, \textup{(B)} and \textup{(C)}.
Then for any $r\geq1$ there exists a finite constant
$B(r)$ such that for any $N\geq1$, $n\geq0$, and $\varphi\in\mathcal
{L}$,
%
\begin{equation}
\mathbb{E} \bigl[\bigl\llvert \pi_{n}^{N}(\varphi)-
\pi_{n}(\varphi)\bigr\rrvert ^{r} \bigr]^{1/r}
\leq4B(r)^{1/r}\sum_{p=0}^{n}
\delta_{p,n}\bigl\llVert P_{p,n}(\bar{\varphi})\bigr\rrVert
\mathbb{E} \bigl[\bigl\llvert \mathcal {C}_{p,n}^{N}\bigr
\rrvert ^{r} \bigr]^{1/r}.\label{eq:L_p_bound_mixing_statement}
\end{equation}
where $\bar{\varphi}:=\varphi-\pi_{n}(\varphi)$ and
\[
\mathcal{C}_{p,n}^{N}:=\frac{\sqrt{\sum_{i} (\beta
_{p,n}^{i}W_{p}^{i} )^{2}}}{\sum_{i}\beta_{p,n}^{i}W_{p}^{i}}.
\]
\end{prop}

\begin{pf*}{Proof of Theorem~\ref{thmm:L_R_mix}} For the first bound on the
right of (\ref{eq:them_Stability_Statement}) under the conditions
of the theorem, we apply Proposition~\ref{prop:norm_const_bound}
to give the following recursive bound:
%
\begin{equation}
v_{n}\leq\sum_{p=0}^{n-1}
\frac{C}{N\tau} (v_{p}+1 ),\label
{eq:v_n_recursion}
\end{equation}
where $v_{n}:=\mathbb{E} [ (Z_{n}^{N}/Z_{n}-1 )^{2} ]$
and
\[
C:=\sup_{p,n}\operatorname{osc} \bigl(\overline{Q}_{p,n}(1)
\bigr)^{2}\leq4\sup_{p,n}\bigl\llVert
\overline{Q}_{p,n}(1)\bigr\rrVert ^{2}<+\infty,
\]
under (C); see (\ref{eq:Q_p,n_bounded}). We shall now
prove
%
\begin{equation}
v_{n}\leq \biggl(1+\frac{C}{N\tau} \biggr)^{n}-1\qquad
\forall n\geq0,\label
{eq:v_n_induction_hyp}
\end{equation}
which holds trivially if $C=0$, since in that case $v_{n}=0$ by
(\ref{eq:v_n_recursion}). Therefore suppose $C>0$. The argument
is inductive. To initialize, note that since by definition $Z_{0}^{N}=Z_{0}=1$,
we have $v_{0}=0$. Now assume (\ref{eq:v_n_induction_hyp}) holds
at all ranks strictly less than some fixed $n\geq1$. Using (\ref
{eq:v_n_recursion}),
we then have at rank $n$,
\begin{eqnarray*}
v_{n} & \leq& \frac{C}{N\tau}\sum_{p=0}^{n-1}
(v_{p}+1 )
\leq \frac{C}{N\tau}\sum_{p=0}^{n-1}
\biggl(1+\frac{C}{N\tau} \biggr)^{p}
\\
& = & \frac{C}{N\tau}\frac{ (1+{C}/{N\tau} )^{n}-1}{
(1+{C}/{N\tau} )-1}
\\
& = & \biggl(1+\frac{C}{N\tau} \biggr)^{n}-1.
\end{eqnarray*}
This completes the proof of (\ref{eq:v_n_induction_hyp}), from which
the first inequality on the right of (\ref{eq:them_Stability_Statement})
follows immediately upon noting that by Theorem~\ref{thmm:convergence},
$\mathbb{E}[Z_{n}^{N}]=Z_{n}$.

For the second bound on the right of (\ref{eq:them_Stability_Statement}),
first note that by Lemma~\ref{lem:Phi_stable}, under (C) we have
\begin{eqnarray*}
\bigl\llVert P_{p,n}(\bar{\varphi})\bigr\rrVert & = & \sup
_{x}\bigl\llvert P_{p,n}(\varphi) (x)-
\pi_{n}(\varphi)\bigr\rrvert
\\
& = & \sup_{x}\bigl\llvert \Phi_{p,n}(
\delta_{x}) (\varphi)-\Phi_{p,n}(\pi _{p}) (
\varphi)\bigr\rrvert
\\
& \leq& \sup_{\mu,\nu\in\mathcal{P}}\bigl\llVert \Phi_{p,n}(\mu)-\Phi
_{p,n}(\nu)\bigr\rrVert _{\mathrm{tv}}\llVert \varphi\rrVert \leq
\llVert \varphi\rrVert C\rho^{n-p},
\end{eqnarray*}
and by (\ref{eq:Q_p,n_bounded}),
\[
\sup_{n\geq0}\sup_{p\leq n} \delta_{p,n}<+
\infty.
\]
Using these upper bounds, the fact that under (B$^{++}$)
we have $\beta_{p,n}^{i}=1/N$, and Proposition~\ref{prop:L_p_bound_mixing},
we find that there exists a finite constant $\widetilde{B}(r)$ such
that for any $N\geq1$, $n\geq0$, $\varphi\in\mathcal{L}$,
\[
\mathbb{E} \bigl[\bigl\llvert \pi_{n}^{N}(\varphi)-
\pi_{n}(\varphi)\bigr\rrvert ^{r} \bigr]^{1/r}\leq
\llVert \varphi\rrVert \frac{\widetilde
{B}(r)}{\sqrt{N}}\sum_{p=0}^{n}
\rho^{n-p}\mathbb{E} \bigl[\bigl\llvert \mathcal {E}_{p}^{N}
\bigr\rrvert ^{-r/2} \bigr]^{1/r},
\]
where
\[
\mathcal{E}_{n}^{N}=\frac{ (N^{-1}\sum_{i}W_{n}^{i}
)^{2}}{N^{-1}\sum_{i} (W_{n}^{i} )^{2}}.\vspace*{-10pt}
\]
\end{pf*}

\section{Discussion}
\label{sec:Discussion}

\subsection{Why not just run independent particle filters and
average?}
\label{sub:Why-not-just}

One obvious approach to parallelization of SMC is to run a number
of independent copies of a standard algorithm, such as the BPF, and
then in some sense simply average their outputs. Let us explain possible
shortcomings of this approach.

Suppose we want to run $s\geq1$ independent copies of Algorithm \ref
{alg:boot_pf},
each with $q\geq1$ particles. For purposes of exposition, it is helpful
to express this collection of independent algorithms as a particular
instance of $\alpha$SMC: for the remainder of Section~\ref{sub:Why-not-just},
we set $N=sq$ and consider Algorithm \ref{alg:aSMC} with $\mathbb{A}_{N}$
chosen to consist only of the block diagonal matrix:
%
\begin{equation}
\left[ %
\matrix{ \mathbf{q^{-1}} &
\mathbf{0} & \cdots& \mathbf{0}
\vspace*{2pt}\cr
\mathbf{0} & \mathbf{q^{-1}} & \cdots& \mathbf{0}
\vspace*{2pt}\cr
\vdots& \vdots& \ddots& \vdots
\vspace*{2pt}\cr
\mathbf{0} & \mathbf{0} & \cdots& \mathbf{q^{-1}} }\right],\label{eq:alpha_block}
\end{equation}
where $\mathbf{q}^{-1}$ is a $q\times q$ submatrix with every entry
equal to $q^{-1}$ and $\mathbf{0}$ is a submatrix of zeros, of the
same size. In this situation, a simple application of Lemma~\ref
{lem:W_n_representation}
shows that for any $n\geq1$ and $\ell\in[s]$, if we define $B(\ell):=\{
(\ell-1)q+1,(\ell-1)q+2,\ldots,\ell q\}$,
then
%
\begin{equation}
\mbox{for all } i_{n}\in B(\ell),\qquad W_{n}^{i_{n}}=
\prod_{p=0}^{n-1} \biggl(N^{-1}\sum
_{i_{p}\in B(\ell)}g_{p} \bigl(\zeta
_{p}^{i_{p}} \bigr) \biggr)=:\mathbb{W}_{n}^{\ell},\label{eq:W_n^i_blocks}
\end{equation}
cf. (\ref{eq:bootstrap_W_n^i})--(\ref{eq:bootstrap_Z_n^N}), and
furthermore upon inspection of Algorithm \ref{alg:aSMC}, we find
%
\begin{equation}
\mbox{for all }\ell\in[s]\mbox{ and }i\in B(\ell)\qquad\mathbb {P} \bigl(
\zeta_{n}^{i}\in A|\mathcal{F}_{n-1} \bigr)=
\frac
{\sum_{j\in B(\ell)}g_{n-1}(\zeta_{n-1}^{j})f(\zeta_{n-1}^{j},A)}{\sum_{j\in B(\ell)}g_{n-1} (\zeta_{n-1}^{j} )},\label{eq:blocks_law}
\end{equation}
for any $A\in\mathcal{X}$. It follows that the blocks of particles
\[
\hat{\zeta}_{n}^{k}:= \bigl\{ \zeta_{n}^{i}
\bigr\} _{i\in B(\ell)},\qquad \ell\in[s],
\]
are independent, and for each $\ell\in[s]$, the sequence $ \{ \hat
{\zeta}_{n}^{\ell};n\geq0 \} $
evolves under the same law as a BPF, with $q$ particles. Furthermore,
we notice
\[
\pi_{n}^{N}=\pi_{n}^{sq}=
\frac{\sum_{i}W_{n}^{i} \delta_{\zeta
_{n}^{i}}}{\sum_{i}W_{n}^{i}}=\frac{\sum_{\ell\in[s]}\sum_{i\in B(\ell
)}W_{n}^{i} \delta_{\zeta_{n}^{i}}}{\sum_{\ell\in[s]}\sum_{i\in B(\ell
)}W_{n}^{i}}=\frac{\sum_{\ell\in[s]}\mathbb{W}_{n}^{\ell}
(q^{-1}\sum_{i\in B(\ell)}\delta_{\zeta_{n}^{i}} )}{\sum_{\ell\in
[s]}\mathbb{W}_{n}^{\ell}},
\]
where $q^{-1}\sum_{i\in B(\ell)}\delta_{\zeta_{n}^{i}}$ may be regarded
as the approximation of $\pi_{n}$ obtained from the $\ell$th block
of particles. Since we have assumed that $\mathbb{A}_{N}$ consists
only of the matrix (\ref{eq:alpha_block}), \textup{(A2)} and (B$^{++}$)
hold, and by Theorem~\ref{thmm:convergence} we are assured of the
a.s. convergence $\pi_{n}^{sq}(\varphi)\rightarrow\pi_{n}(\varphi)$
when $q$ is fixed and $s\rightarrow\infty$. In words, we have convergence
as the total number of bootstrap algorithms tends to infinity, even
though the number of particles within each algorithm is fixed. On
the other hand, simple averaging of the output from the $s$ independent
algorithms would entail reporting:
%
\begin{equation}
\frac{1}{sq}\sum_{i\in[sq]}\delta_{\zeta_{n}^{i}}\label{eq:naive}
\end{equation}
as an approximation of $\pi_{n}$; the problem is that (\ref{eq:naive})
is biased, in the sense that in general it is not true that, with
$q$ fixed, $(sq)^{-1}\sum_{i\in[sq]}\varphi(\zeta_{n}^{i})\rightarrow\pi
_{n}(\varphi)$
as $s\rightarrow\infty$ (although obviously we do have convergence
if $q\rightarrow\infty$). In summary, simple averages across independent
particle filters do not, in general, converge as the number of algorithms
grows.

We can also discuss the quality of an approximation of $Z_{n}$ obtained
by simple averaging across the $s$ independent algorithms; let us
consider the quantities
\[
\mathbb{Z}_{n}^{(q,\ell)}:=\frac{1}{\ell}\sum
_{j\in[\ell]}\mathbb {W}_{n}^{j}, \qquad\ell\in[s].
\]
Comparing (\ref{eq:W_n^i_blocks}) with (\ref{eq:bootstrap_Z_n^N}),
and noting (\ref{eq:blocks_law}) and the independence properties
described above, we have
%
\begin{equation}
\mathbb{E} \bigl[\mathbb{Z}_{n}^{(q,s)} \bigr]=Z_{n},\qquad
\mathbb {E} \biggl[ \biggl(\frac{\mathbb{Z}_{n}^{(q,s)}}{Z_{n}}-1 \biggr)^{2} \biggr]=
\frac{1}{s}\mathbb{E} \biggl[ \biggl(\frac{\mathbb
{Z}_{n}^{(q,1)}}{Z_{n}}-1
\biggr)^{2} \biggr],\label{eq:Z_naive_average}
\end{equation}
where the first equality holds due to the first part of Theorem~\ref
{thmm:convergence}:
in this context the well known lack-of-bias property of the BPF. Under
certain ergodicity and regularity conditions $\mathbb{E} [
(\mathbb{Z}_{n}^{(q,1)}/Z_{n} )^{2} ]$
can grow exponentially fast along observation sample paths when $q$
is fixed and $n\rightarrow\infty$ \cite{WhiteleyTPF}. When that
occurs, it is clear from (\ref{eq:naive}) that $s$ must be scaled
up exponentially fast with $n$ in order to control the relative variance
of $\mathbb{Z}_{n}^{(q,s)}$. On the other hand, by Theorem~\ref{thmm:L_R_mix}
and Remark~\ref{rem:linear_variance}, it is apparent that if we
design an instance of $\alpha$SMC so as to enforce $\inf_{n\geq
0}\mathcal{E}_{n}^{N}>0$,
then we can control $\mathbb{E} [ (Z_{n}^{N}/Z_{n}
)^{2} ]$
at a more modest computational cost. When $\mathbb{A}_{N}$ consists
only of the matrix (\ref{eq:alpha_block}) we do not have a guarantee
that $\inf_{n\geq0}\mathcal{E}_{n}^{N}>0$, but in Section~\ref{sub:Algorithms-with-adaptive}
we shall suggest some novel algorithms which do guarantee this lower
bound and therefore enjoy the time-uniform convergence and linear-in-time
variance properties of Theorem~\ref{thmm:L_R_mix}. Before addressing
these stability issues, we discuss the conditions under which the
$\alpha$SMC
algorithm converges.

\begin{figure}

\includegraphics{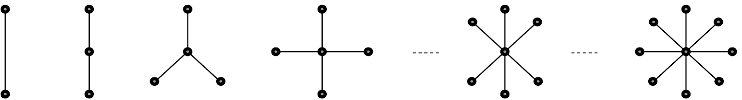}

\caption{Star graphs.} \label{fig:Star-graphs-of}
\end{figure}

\subsection{Ensuring convergence}
\label{sub:Ensuring-convergence}

Throughout Section~\ref{sub:Ensuring-convergence}, we consider the
generic Algorithm \ref{alg:aSMC}. We describe what can go wrong
if the conditions (B$^{+}$) and (B$^{++}$) of
Theorem~\ref{thmm:convergence} do not hold: suppose that $\mathbb{A}_{N}$
consists only of the transition matrix of a simple random walk on
the star graph with $N$ vertices, call it $\mathcal{S}_{N}$. That
is, for $N>2$, $\mathcal{S}_{N}$ is an undirected tree with one
internal vertex and $N-1$ leaves, and for $N\leq2$, all vertices
are leaves. Examples of $\mathcal{S}_{N}$ are illustrated in Figure~\ref{fig:Star-graphs-of}.
It is elementary that a simple random walk on $\mathcal{S}_{N}$ has
unique invariant distribution given by
\[
\frac{d_{N}^{i}}{\sum_{j}d_{N}^{j}},\qquad i\in[N], \mbox{ where } d_{N}^{i}:=
\mbox{ degree of vertex }i\mbox{ in }\mathcal{S}_{N},
\]
so that (B$^{++}$) does not hold for $N>2$. Assuming that
for every $N>2$, the internal vertex of $\mathcal{S}_{N}$ is labelled
vertex $1$, then elementary calculations show that
\[
\max_{i\in[N]}\beta_{p,n}^{i}=
\beta_{p,n}^{1}= %
\cases{ \displaystyle\frac{N-1}{N}, &\quad
$\mbox{if $(n-p)$ is even}$,\vspace*{2pt}
\cr
\displaystyle\frac{1}{N}, &\quad  $\mbox{if $(n-p)$ is odd}$,}
\]
so (B$^{+}$) also does not hold, and thus part (2) of
Theorem~\ref{thmm:convergence}
does not hold.

As a more explicit example of convergence failure, suppose that $\mathbb{A}_{N}$
consists only of the matrix which has $1$ for every entry in its
first column, and zeros for all other entries. This is the transition
matrix of a random walk on a directed graph of which all edges lead
to vertex $1$. It follows that for all $0\leq p<n$, we have $\beta_{p,n}^{1}=1$
and $\beta_{p,n}^{i}=0$ for all $i\in[N]\setminus\{1\}$, so (B$^{+}$)
clearly does not hold. If additionally $f(x,\cdot)=\delta_{x}(\cdot)$,
then by inspection of Algorithm \ref{alg:aSMC} we have $\mathbb
{P} ( \{ \zeta_{n}^{i}=\zeta_{0}^{1} \}  )=1$
for all $i\in[N]$ and all $n\geq1$. We then also have $\mathbb{P}
( \{ \pi_{n}^{N}=\delta_{\zeta_{0}^{1}} \}  )=1$,
so that we obtain a generally poor and non-convergent approximation
of $\pi_{n}$.

In both these situations vertex $1$ is, in graph theoretic terms,
a \emph{hub} and an intuitive explanation of the convergence failure
is that the contribution of particle $1$ to $\pi_{n}^{N}$ does not
become negligible as $N\rightarrow\infty$, so that no ``averaging''
takes place. Assumption (B$^{+}$) ensures enough negligibility
to prove the weak laws of large numbers in Theorem~\ref{thmm:convergence}.
Assumption (B$^{++}$) may be viewed as ensuring negligibility,
and in such a way as to ensure the $\sqrt{N}$ rate of convergence
and strong law in the final part of Theorem~\ref{thmm:convergence}.
As a practical summary, we recommend verifying (B$^{++}$),
or at least avoid using graphs with hubs, since otherwise
$\alpha$SMC may not converge.

\subsection{Provably stable algorithms with adaptive interaction}
\label{sub:Algorithms-with-adaptive}

There are of course many choices of $\mathbb{A}_{N}$ which do satisfy
(B$^{++}$). In this section, we provide some guidance and
suggestions on this matter. In order to focus our attention, we consider
in addition to (B$^{++}$), the following criteria against
which to assess candidates for $\mathbb{A}_{N}$ and whatever functional
is used at line $(\star)$ of Algorithm \ref{alg:aSMC}:
\begin{longlist}[(a)]

\item[(a)]
the condition $\inf_{n\geq0}\mathcal{E}_{n}^{N}>0$
should be enforced, so as to ensure stability,

\item[(b)]
the computational complexity of associated
sampling, weight and ESS calculations should not be prohibitively
high.
\end{longlist}
The motivation for (a) is the
theoretical assurance given by Theorem~\ref{thmm:L_R_mix}. The motivation
for (b) is simply that we do not want an algorithm
which is much more expensive than any of the standard SMC methods,
Algorithms \ref{alg:SIS}--\ref{alg:boot_pf} and the ARPF. It is
easily checked that the complexity of SIS is $\mathrm{O}(N)$ per unit time
step, which is the same as the complexity of the BPF \cite
{carpenter1999improved}
and the ARPF.

Throughout the remainder of Section~\ref{sub:Algorithms-with-adaptive},
we shall assume that $\mathbb{A}_{N}$ consists only of transition
matrices of simple random walks on regular undirected graphs. We impose
a little structure in addition to this as per the following definition,
which identifies an object related to the standard notion of a block-diagonal
matrix.
\begin{defn*}
A \emph{B-matrix} is a Markov transition matrix which specifies
a simple random walk on a regular undirected graph which has a self-loop
at every vertex and whose connected components are all complete subgraphs.
\end{defn*}
Note that due to the graph regularity appearing in this definition,
if $\mathbb{A}_{N}$ consists only of B-matrices, then (B$^{++}$)
is immediately satisfied. This regularity is also convenient for purposes
of interpretation: it seems natural to use graph degree to give a
precise meaning to ``degree of interaction''. Indeed $\mathrm{Id}$ and $\mathbf{1}_{1/N}$
are both B-matrices, respectively, specifying simple random walks on
$1$-regular and $N$-regular graphs, and recall for the ARPF, $\mathbb
{A}_{N}= \{ \mathrm{Id},\mathbf{1}_{1/N} \} $;
the main idea behind the new algorithms below is to consider an instance
of $\alpha$SMC in which $\mathbb{A}_{N}$ is defined to consist of
B-matrices of various degrees $d\in[N]$, and define adaptive algorithms
which select the value of $\alpha_{n-1}$ by searching through $\mathbb{A}_{N}$
to find the graph with the smallest $d$ which achieves $\mathcal
{E}_{n}^{N}\geq\tau>0$
and hence satisfies criterion (a). In this way, we ensure provable
stability whilst trying to avoid the complete interaction which occurs
when $\alpha_{n-1}=\mathbf{1}_{1/N}$.

Another appealing property of B-matrices is formalized in the following
lemma; see criterion~(b) above. The proof is given
in the \hyperref[sec:Appendix]{Appendix}.

\begin{lem}
\label{lem:complexity}Suppose that $A= (A^{ij} )$ is a
B-matrix of size $N$. Then given the quantities $ \{
W_{n-1}^{i} \} _{i\in[N]}$
and $ \{ g_{n-1}(\zeta_{n-1}^{i}) \} _{i\in[N]}$, the computational
complexity of calculating $ \{ W_{n}^{i} \} _{i\in[N]}$
and simulating $ \{ \zeta_{n}^{i} \} _{i\in[N]}$ as per
Algorithm \ref{alg:aSMC}, using $\alpha_{n-1}=A$, is $\mathrm{O}(N)$.
\end{lem}

When calculating the overall complexity of Algorithm \ref{alg:aSMC}
we must also consider the complexity of line $(\star)$, which in
general depends on $\mathbb{A}_{N}$ and the particular functional
used to choose $\alpha_{n}$. We resume this complexity discussion
after describing the specifics of some adaptive algorithms.

\textit{Adaptive interaction}.

Throughout this section, we set $m\in\mathbb{N}$ and then $N=2^{m}$.
Consider Algorithm \ref{alg:aSMC} with $\mathbb{A}_{N}$ chosen
to be the set of B-matrices of size $N$. We suggest three adaptation
rules at line $(\star)$ of Algorithm \ref{alg:aSMC}: Simple, Random,
and Greedy, all implemented via Algorithm \ref{alg:generic adaptation}
(note that dependence of some quantities on $n$ is suppressed from
the notation there), but differing in the way they select the index
list $\mathcal{I}_{k}$ which appears in the ``while'' loop of that
procedure. The methods for selecting $\mathcal{I}_{k}$ are summarised
in Table~\ref{tab:Choosing_I_k}: the Simple rule needs little explanation,
the Random rule implements an independent random shuffling of indices
and the Greedy rule is intended, heuristically, to pair large weights,
$\mathbb{W}_{k}^{i}$, with small weights in order to terminate the
``while'' loop with as small a value of $k$ as possible. Note that,
formally, in order for our results for $\alpha$SMC to apply when
the Random rule is used, the underlying probability space must be
appropriately extended, but the details are trivial so we omit them.

\begin{algorithm}
\begin{raggedright}
\quad At iteration $n$ and line $(\star)$ of Algorithm \ref{alg:aSMC}
\end{raggedright}

\begin{raggedright}
\qquad For $i=1,\ldots,N$,
\end{raggedright}

\begin{raggedright}
\qquad\quad Set $B(0,i)=\{i\}$, $\mathbb
{W}_{0}^{i}=W_{n-1}^{i}g_{n-1}(\zeta_{n-1}^{i})$,
\end{raggedright}

\begin{raggedright}
\qquad Set $k=0$,
\end{raggedright}

\begin{raggedright}
\qquad Set $\overline{\mathbb{W}}_{0}=N^{-1}\sum_{i}\mathbb{W}_{0}^{i}$
, $\mathcal{E}=\frac{ (\overline{\mathbb{W}}_{0}
)^{2}}{N^{-1}\sum_{i} (\mathbb{W}_{0}^{i} )^{2}}$,
\end{raggedright}

\begin{raggedright}
\qquad While $\mathcal{E}<\tau$
\end{raggedright}

\begin{raggedright}
\qquad\quad Set $\mathcal{I}_{k}$ according to the
Simple, Random or Greedy scheme of Table~\ref{tab:Choosing_I_k}
\end{raggedright}

\begin{raggedright}
\qquad\quad For $i=1,\ldots,N/2^{k+1}$
\end{raggedright}

\begin{raggedright}
\qquad\qquad Set $B(k+1,i)=B(k,\mathcal
{I}_{k}(2i-1))\cup B(k,\mathcal{I}_{k}(2i))$
\end{raggedright}

\begin{raggedright}
\qquad\qquad Set $\mathbb{W}_{k+1}^{i}=\mathbb
{W}_{k}^{\mathcal{I}_{k}(2i-1)}/2+\mathbb{W}_{k}^{\mathcal{I}_{k}(2i)}/2$
\end{raggedright}

\begin{raggedright}
\qquad\quad Set $k=k+1$
\end{raggedright}

\begin{raggedright}
\qquad\quad Set $\mathcal{E}=\frac{ (\overline{\mathbb
{W}}_{0} )^{2}}{N^{-1}2^{k}\sum_{i\in[N/2^{k}]} (\mathbb
{W}_{k}^{i} )^{2}}$
\end{raggedright}

\begin{raggedright}
\qquad Set $K_{n-1}=k$
\end{raggedright}

\begin{raggedright}
\qquad Set $\alpha_{n-1}^{ij}=
\cases{
1/2^{K_{n-1}}, &\quad  $\mbox{if }i\sim j\mbox{ according to
$\{ B(K_{n-1},i) \} _{i\in[N/2^{K_{n-1}}]}$},$\vspace*{2pt}\cr
0, &\quad  $\mbox{otherwise}.$}
$
\end{raggedright}
\caption{Adaptive selection of
$\alpha_{n-1}$}
\label{alg:generic adaptation}
\end{algorithm}

Following the termination of the ``while'' loop, Algorithm \ref
{alg:generic adaptation}
outputs an integer $K_{n-1}$ and a partition $ \{ B(K_{n-1},i)
\} _{i\in[N/2^{K_{n-1}}]}$
of $[N]$ into $N/2^{K_{n-1}}$ subsets, each of cardinality $2^{K_{n-1}}$;
this partition specifies $\alpha_{n-1}$ as a B-matrix and $2^{K_{n-1}}$
is the degree of the corresponding graph (we keep track of $K_{n-1}$
for purposes of monitoring algorithm performance in Section~\ref{sub:Numerical-illustrations}).
Proposition~\ref{prop:Upon-termination-of} is a formal statement
of its operation and completes our complexity considerations. The
proof is given in the \hyperref[sec:Appendix]{Appendix}. It can be checked by an inductive
argument similar to the proof of Lemma~\ref{lem:ARPF_A2}, also in
the \hyperref[sec:Appendix]{Appendix}, that when $\alpha_{n}$ is chosen according to
Algorithm \ref{alg:generic adaptation}
combined with any of the adaptation rules in Table~\ref{tab:Choosing_I_k},
\textup{(A2)} is satisfied.

\begin{prop}
\label{prop:Upon-termination-of}The weights $ \{ \mathbb
{W}_{k}^{i} \} _{i\in[N/2^{k}]}$
calculated in Algorithm \ref{alg:generic adaptation} obey the expression
%
\begin{equation}
\mathbb{W}_{k}^{i}=2^{-k}\sum
_{j\in B(k,i)}W_{n-1}^{j}g_{n-1}\bigl(
\zeta _{n-1}^{j}\bigr).\label{eq:W_k_explicit}
\end{equation}
Moreover, $\alpha_{n-1}$ delivered by Algorithm \ref{alg:generic adaptation}
is a B-matrix and when this procedure is used at line $(\star)$ of
Algorithm \ref{alg:aSMC}, the weights calculated in Algorithm \ref{alg:aSMC}
are given, for any $i\in[N/2^{K_{n-1}}]$, by
%
\begin{equation}
W_{n}^{j}=\mathbb{W}_{K_{n-1}}^{i}\qquad
\mbox{for all }j\in B(K_{n-1},i)\label{eq:W_equals_bb_W}
\end{equation}
and $\mathcal{E}_{n}^{N}\geq\tau$ always. The overall worst-case
complexity of Algorithm \ref{alg:aSMC} is, for the three adaptation
rules in Table~\ref{tab:Choosing_I_k}, Simple: $\mathrm{O}(N)$, Random:
$\mathrm{O}(N)$, and Greedy: $\mathrm{O}(N\log_{2}N)$.
\end{prop}

\begin{table}
\caption{Adaptation rules for choosing
$\mathcal{I}_{k}$}\label{tab:Choosing_I_k}
\begin{tabular*}{\textwidth}{@{\extracolsep{\fill}}lp{330pt}@{}}
\hline
Simple & set $\mathcal
{I}_{k}=(1,\ldots,N/2^{k})$\\
Random &
if $k=0$, set $\mathcal{I}_{k}$ to a random permutation
of $[N/2^{k}]$, otherwise $\mathcal{I}_{k}=(1,\ldots,
,N/2^{k})$\\
Greedy &
set $\mathcal{I}_{k}$ such that $\mathbb{W}_{k}^{\mathcal
{I}_{k}(1)}\geq\mathbb{W}_{k}^{\mathcal{I}_{k}(3)}\geq\cdots\geq\mathbb
{W}_{k}^{\mathcal{I}_{k}(N/2^{k}-1)}\geq\mathcal{\mathbb
{W}}_{k}^{\mathcal{I}_{k}(N/2^{k})}\geq\cdots\geq\mathbb{W}_{k}^{\mathcal
{I}_{k}(4)}\geq\mathbb{W}_{k}^{\mathcal{I}_{k}(2)}$\\
\hline
\end{tabular*}
\end{table}

\subsection{Numerical illustrations}
\label{sub:Numerical-illustrations}

We consider a stochastic volatility HMM:
\begin{eqnarray*}
X_{0}&\sim&\mathcal{N}(0,1),\qquad X_{n}=aX_{n-1}+
\sigma V_{n},
\\
 Y_{n}&=&\varepsilon W_{n}\exp(X_{n}/2),
\end{eqnarray*}
where $ \{ V_{n} \} _{n\in\mathbb{N}}$ and $ \{ W_{n}
\} _{n\in\mathbb{N}}$
are sequences of mutually i.i.d. $\mathcal{N}(0,1)$ random variables,
$\llvert a\rrvert <1$, and $\sigma,\varepsilon>0$. To study the behaviour
of the different adaptation rules in terms of effective sample size,
a sequence of $3\cdot10^{4}$ observations were generated from the
model with $a=0.9$, $\sigma=0.25$, and $\varepsilon=0.1$. This
model obviously does not satisfy (C), but (A1)
is satisfied as long as the observation record does not include the
value zero.

\begin{figure}

\includegraphics{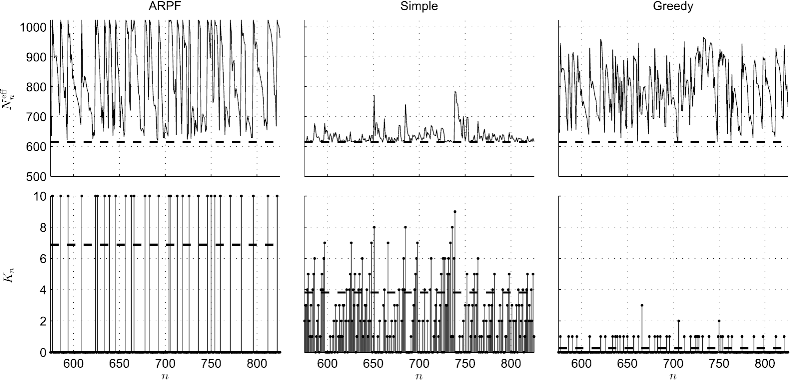}

\caption{Snapshots of ESS and degree of
interaction.
Top: $N_{n}^{\mathrm{eff}}$ vs. $n$ (solid) and threshold $\tau N$
(dashed). Bottom: $K_{n}$ vs. $n$ (stems) and the base two logarithm
of the time-average of $2^{K_{n}}$ (dashed). Recall from Section \protect\ref{sub:Algorithms-with-adaptive} that $2^{K_{n}}$ is the degree
of the graph corresponding to the matrix $\alpha_{n}$ selected by
Algorithm \protect\ref{alg:generic adaptation}, and returned to line $(\star)$
of Algorithm \protect\ref{alg:aSMC}.}\label{fig:ESS-and-interaction}
\end{figure}

The ARPF and $\alpha$SMC with the Simple, Random and Greedy adaptation
procedures specified in Section~\ref{sub:Algorithms-with-adaptive}
were run on this data with $N=2^{10}$ and threshold $\tau=0.6$.
To give some impression of ESS and interaction behaviour, Figure~\ref{fig:ESS-and-interaction}
shows snapshots of $N_{n}^{\mathrm{eff}}$ and $K_{n}$ versus $n$,
for $575\leq n\leq825$. The sample path of $N_{n}^{\mathrm{eff}}$
for ARPF displays a familiar saw-tooth pattern, jumping back up to
$N=2^{10}$ when resampling, that is, when $K_{n}=10$. The Simple adaptation
scheme keeps $N_{n}^{\mathrm{eff}}$ just above the threshold $\tau
N=0.6\times2^{10}$,
whereas the Greedy strategy is often able to keep $N_{n}^{\mathrm{eff}}$
well above this threshold, with smaller values of $K_{n}$, that is, with
a lower degree of interaction. The results for the Random adaptation
rule, not shown in this plot, where qualitatively similar to those
of the Greedy algorithm but slightly closer to the Simple adaptation.

In order to examine the stationarity of the particle processes as
well as the statistical behavior of the degree of interaction over
time, Figure~\ref{fig:histograms_and_E_vs_k} shows two histograms
of $K_{n}$ for each of the adaptation rules. One histogram is based
on the sample of $K_{n}$ where $100<n\leq15\,050,$ and the other is
based on $K_{n}$ where $15\,050<n\leq30\,000$. For each algorithm, the
similarity between the histograms for the two time intervals suggests
that the process $ \{ K_{n} \} _{n\geq0}$ is stationary.
As expected, the distribution of $K_{n}$ for ARPF is dichotomous
taking only values equal to $K_{n}=0$ when there is no interaction,
that is, the resampling is skipped or $K_{n}=10$ for the complete interaction,
that is, resampling. It is apparent that the Simple, Random and Greedy
algorithms move the distribution of $K_{n}$ toward smaller values
and almost always manage to avoid the complete interaction. For the
Random and Greedy algorithms, $K_{n}$ rarely exceeds $1$, that is, in
order to guarantee $\mathcal{E}_{n}^{N}$ it is rarely necessary to
consider anything more than pair-wise interaction.

The plot on the right of Figure~\ref{fig:histograms_and_E_vs_k}
shows, for each of the Simple, Random and Greedy adaptation rules,
the relationship between the intermediate variables $\mathcal{E}$
and $k$ appearing in the ``while'' loop of Algorithm \ref{alg:generic
adaptation}.
In order to obtain equal sample sizes for plotting purposes,
Algorithm \ref{alg:generic adaptation}
was modified slightly so as to evaluate $\mathcal{E}$ for every value
$k\in\{0,\ldots,m\}$, whilst still outputting $K_{n-1}$ as the smallest
value of $k$ achieving $\mathcal{E}\geq\tau$. The plotted data were
then obtained, for each $k$, by averaging the corresponding values
of $\mathcal{E}$ over the time steps of the algorithm. It is apparent
that, for small values of $k$, the Random and Greedy strategies achieve
a faster increase in $\mathcal{E}$ than the Simple strategy, and
this explains the shape of the histograms on the left of Figure~\ref{fig:histograms_and_E_vs_k}.
%
\begin{figure}

\includegraphics{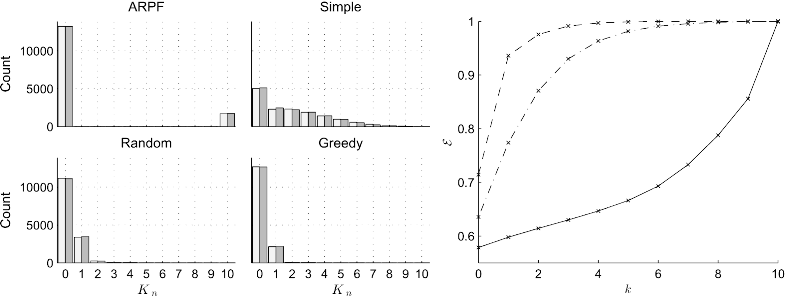}

\caption{Left: Histograms of
$K_{n}$ for
the ARPF and the three adaptation rules of Table \protect\ref{tab:Choosing_I_k}.
The light bars were obtained from $ \{ K_{n};n=101,\ldots
,15\,050 \} $
and the dark bars from $ \{ K_{n};n=15\,051,\ldots,30\,000 \} $
Right: Growth of $\mathcal{E}$ vs. $k$ for the Simple (solid),
Random (dash-dot) and Greedy (dashed).}
\label{fig:histograms_and_E_vs_k}
\end{figure}

Figure~\ref{fig:Mean-square-error} shows a comparison of the mean
squared errors (MSE) of approximating the conditional expectation
of $\phi(X_{p})$ with respect to the underlying stochastic volatility
HMM given the observations $\{y_{n};0\leq n\leq p+\ell\}$, where
$\ell\in\{-5,0,1\}$ and $\phi$ is some test function. The cases,
$\ell=-5$, $\ell=0$, and $\ell=1$ correspond to the lag 5 smoother,
filter and one step predictor, respectively. The lag 5 smoother results
were obtained by tracing back ancestral lineages. In order to estimate
the approximation error, a reference value for the conditional expectation
was evaluated by running a BPF with a large sample size $N=2^{17}$.
Approximation errors were evaluated for $N_{\mathrm{MC}}=1000$ Monte
Carlo runs of 1000 time steps each with $N=2^{9}$, and MSE was obtained
by averaging over the time steps and the Monte Carlo runs. First 30
time steps were excluded in the calculations to avoid any non-stationary
effects due to initialization. The results show that the Random and
Greedy algorithms produce consistently smaller errors than the Simple
algorithm and for large values of $\tau$ the Greedy algorithm appears
to consistently outperform ARPF.

\begin{figure}

\includegraphics{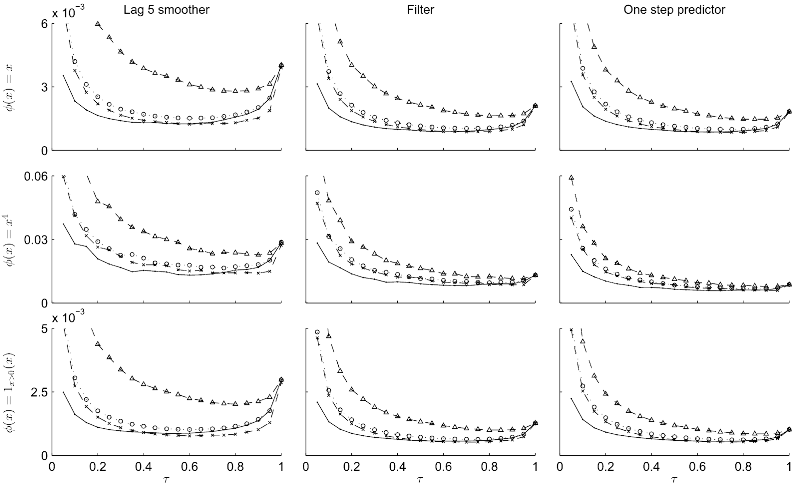}

\caption{MSE vs. $\tau$ for the lag 5 smoother,
filter, and one step predictor using the four algorithms ARPF (solid),
Simple ($\triangle$), Random ($\circ$), and Greedy ($\times)$ and
three test functions $\phi$.} \label{fig:Mean-square-error}
\end{figure}

\subsection{Concluding remarks}
\begin{itemize}
\item The martingale decomposition presented in Proposition~\ref
{prop:martingale}
may also be exploited to pursue central limit theorems. A study of
this will be conducted elsewhere, but we believe, further to Remark~\ref
{rem_alpha_conv},
that it will in general involve some further hypotheses in order to
ensure convergence of the covariance of this martingale and thus prove
the existence of a well-defined asymptotic variance.
\item It is worth pointing out that there are also SMC algorithms other
than those listed in Section~\ref{sub:Instances-of-SMC} that can
be formulated as instances of $\alpha$SMC, for example, the stratified
resampling
algorithm of Kitagawa \cite{Kitagawa1996} and the auxiliary particle filter
of Pitt and Shephard \cite{pitt1999filtering}. It should be kept in mind, however,
that the successful formulation of any algorithm as an instance of
$\alpha$SMC does not necessarily imply that the assumptions (B),
(B$^{+}$) or (B$^{++}$) hold, and the validity
of Theorems \ref{thmm:convergence} and \ref{thmm:L_R_mix} is in
that sense, of course, not automatic.
\end{itemize}

\begin{appendix}
\section*{Appendix}
\label{sec:Appendix}

\begin{lem}
\label{lem:ARPF_A2}If for every $n\geq0$, $\alpha_{n}$ is chosen
according to the ESS thresholding rule (\ref{eq:alpha_ARPF}), then~\textup{(A2)} is satisfied.
\end{lem}

\begin{pf}
The proof is by induction. To initialize, we have at rank $n=0$,
%
\renewcommand{\theequation}{\arabic{equation}}
\setcounter{equation}{52}
\begin{equation}
\alpha_{0}:= %
\cases{ \mathbf{1}_{1/N}, &\quad $
\mbox{if }\displaystyle \frac{ (N^{-1}\sum_{i}W_{0}^{i}g_{0}(\zeta_{0}^{i}) )^{2}}{N^{-1}\sum_{i}
(W_{0}^{i}g_{0}(\zeta_{0}^{i}) )^{2}}<\tau$,\vspace*{2pt}
\cr
\mathrm{Id}, &\quad $\mbox{otherwise}$,}
\label{eq:alpha_ARPF-1}
\end{equation}
noting that by definition $W_{0}^{i}=1$, we find that the entries
of $\alpha_{0}$ are measurable w.r.t. $\mathcal{F}_{0}$. For the
induction hypothesis, suppose that for some $n\geq0$ and all $p\leq n$,
the entries of $\alpha_{p}$ are measurable w.r.t. $\mathcal{F}_{n}$.
It follows immediately from Lemma~\ref{lem:W_n_representation},
equation (\ref{eq:unwind2}), that $ \{ W_{n+1}^{i} \} _{i\in[N]}$
are all measurable w.r.t. $\mathcal{F}_{n+1}$, and it follows from
(\ref{eq:alpha_ARPF}) applied at rank $n+1$ that the entries of
$\alpha_{n+1}$ are measurable w.r.t. $\mathcal{F}_{n+1}$, and hence
the induction hypothesis holds at rank $n+1$. This completes the
proof.
\end{pf}

\textit{Resampling times description of the ARPF}.
In order to derive expressions for $\pi_{n}^{N}$ and $Z_{n}^{N}$
in the case of the ARPF, define a family of random sets $ \{ \sigma
_{n};n\geq1 \} $,
and random times $ \{ T_{n};n\geq1 \} $ as follows
%
\renewcommand{\theequation}{\arabic{equation}}
\setcounter{equation}{53}
\begin{eqnarray}\label{eq:T_n_defn}
\sigma_{n} & := & \{ m; 1\leq m\leq n\mbox{ and } \alpha
_{m-1}=\mathbf{1}_{1/N} \},
\nonumber
\\[-8pt]
\\[-8pt]
\nonumber
T_{n} & := & \max (\sigma_{n} ),
\end{eqnarray}
with $T_{n}:=0$ on the event $ \{ \sigma_{n}=\varnothing \} $.
Intuitively, $T_{n}$ can be thought of as the last resampling time
before $n$. Then by construction, using the recursive definition
of $W_{n}^{i}$ in (\ref{eq:W_n_defn}), and (\ref{eq:T_n_defn}),
we have on the event $ \{ \sigma_{n}\neq\varnothing \} $,
%
\begin{eqnarray}\label{eq:W_adapt}
W_{T_{n}}^{i} & =&\sum_{j}
\alpha _{T_{n}-1}^{ij}W_{T_{n}-1}^{j}g_{T_{n}-1}
\bigl(\zeta_{T_{n}-1}^{j}\bigr)
\nonumber
\\[-8pt]
\\[-8pt]
\nonumber
&  =&\frac{1}{N}\sum_{j}W_{T_{n}-1}^{j}g_{T_{n}-1}
\bigl(\zeta _{T_{n}-1}^{j}\bigr) =: \widetilde{W}_{n},\qquad
n\geq1,
\end{eqnarray}
which is independent of $i$. On the event $\{\sigma_{n}=\varnothing\}$,
define $\widetilde{W}_{n}:=1$.

On the event $\{\sigma_{n}\neq\varnothing\}\cap\{T_{n}=n\}$, we trivially
have $W_{n}^{i}=W_{T_{n}}^{i}=\widetilde{W}_{n}$, by (\ref{eq:W_adapt}).
On the event $\{\sigma_{n}\neq\varnothing\}\cap\{T_{n}<n\}$, applying
equation (\ref{eq:unwind}) of Lemma~\ref{lem:W_n_representation}
with $p=T_{n}$, and (\ref{eq:W_adapt}), yields
\begin{eqnarray*}
W_{n}^{i_{n}} & = & \sum_{ (i_{T_{n}},\ldots,i_{n-1} )\in
[N]^{n-T_{n}}}W_{T_{n}}^{i_{T_{n}}}
\prod_{q=T_{n}}^{n-1}g_{q}\bigl(\zeta
_{q}^{i_{q}}\bigr)\alpha_{q}^{i_{q+1}i_{q}}
\\
& = & \widetilde{W}_{n}\sum_{ (i_{T_{n}},\ldots,i_{n-1} )\in
[N]^{n-T_{n}}}\prod
_{q=T_{n}}^{n-1}g_{q}\bigl(
\zeta_{q}^{i_{q}}\bigr)\mathbb {I}[i_{q+1}=i_{q}]
 =  \widetilde{W}_{n}\prod_{p=T_{n}}^{n-1}g_{p}
\bigl(\zeta_{p}^{i_{n}}\bigr).
\end{eqnarray*}
Collecting the above definitions and substituting into (\ref{eq:pi^N_andZ^N})
gives
\[
\pi_{n}^{N}=\frac{\sum_{i}\delta_{\zeta_{n}^{i}}\prod_{p=T_{n}}^{n-1}g_{p}
(\zeta_{p}^{i})}{\sum_{i}\prod_{p=T_{n}}^{n-1}g_{p}(\zeta_{p}^{i})},\qquad Z_{n}^{N}=
\widetilde {W}_{n}\cdot\frac{1}{N}\sum
_{i}\prod_{p=T_{n}}^{n-1}g_{p}
\bigl(\zeta_{p}^{i}\bigr),
\]
with the convention $\prod_{p=n}^{n-1}g_{p}(\zeta_{p}^{i})=1$. Similar
elementary calculations can be used to derive expressions for the
sampling steps of the ARPF, in the interests of brevity we leave it
to the reader to write out the details.

\subsection*{Proofs and auxiliary results for Section \protect\ref
{sec:Martingale-approximations-and}}

The proof of the main result of Section~\ref{sec:Martingale-approximations-and},
Theorem~\ref{thmm:convergence}, hinges on a martingale decomposition
of errors associated with $Z_{n}^{N}$ and $\pi_{n}^{N}(\varphi)$.
This is the subject of Proposition~\ref{prop:martingale}, which we
prove below. Our overall approach is inspired by some of the ideas
of \cite{smc:theory:Dm04}, Chapters 7 and 9, but the path we take
and the details are necessarily different since the analysis of \cite
{smc:theory:Dm04}
does not apply to $\alpha$SMC in general. The following well-known
lemma has been used extensively in the study of sequential Monte Carlo
methods and we shall apply it in the proof of Proposition~\ref
{prop:martingale}.

\begin{lem}[(\cite{smc:theory:Dm04}, Lemma~7.3.3)]
\label{lem:del_moral_moment}
Let $(\mu_{i})_{i\geq1}$ and $(h_{i})_{i\geq1}$ be, respectively,
a sequence of probability measures and a sequence of $\mathbb{R}$-valued,
measurable functions with finite oscillations on a given measurable
space $ (E,\mathcal{E} )$. Assume that $\mu_{i}(h_{i})=0$
for all $i\geq1$ and let $ (X_{i} )_{i\geq1}$ be a sequence
of independent random variables with respective distributions $(\mu
_{i})_{i\geq1}$.
Then for any $r\geq1$,
\[
\sqrt{N}\mathbb{E} \Biggl[\Biggl\llvert N^{-1}\sum
_{i=1}^{N}h_{i}(X_{i})\Biggr
\rrvert ^{r} \Biggr]^{1/r}\leq d(r)^{1/r}\sqrt
{N^{-1}\sum_{i=1}^{N}
\bigl[\operatorname{osc}(h_{i}) \bigr]^{2}},
\]
where for $r\geq1$,
\[
d(2r)=2^{-r}\frac{(2r)!}{r!},\qquad d(2r-1)=\frac
{2^{-(r-1/2)}}{(r-1/2)^{1/2}}
\frac{(2r-1)!}{(r-1)!}.
\]
\end{lem}

\begin{pf*}{Proof of Proposition~\ref{prop:martingale}}
Applying the identities
$\beta_{p-1,n}^{i_{p-1}}=\sum_{i_{p}}\beta_{p,n}^{i_{p}}\alpha
_{p-1}^{i_{p}i_{p-1}}$,
see (\ref{eq:beta_defn}), and $\overline{W}_{p}^{i_{p}}=\sum_{i_{p-1}}\alpha_{p-1}^{i_{p}i_{p-1}}\overline
{W}_{p-1}^{i_{p-1}}\overline{Q}_{p}(1)(\zeta_{p-1}^{i_{p-1}})$,
see (\ref{eq:W_n_defn}), (\ref{eq:Q_bar_defn}), (\ref{eq:W_bar_defn});
with the conventions $\alpha_{-1}:=\mathrm{Id}$, and $\overline{\Gamma
}_{-1}^{N}\overline{Q}_{-1,n}(\varphi)=\overline{W}_{-1}^{i}\overline
{Q}_{-1,n}(\varphi)(\zeta_{-1}^{i}):=\mu_{0}\overline{Q}_{0,n}(\varphi
)=\pi_{n}(\varphi)$,
we have
%
\begin{eqnarray}\label{eq:Gamma-gamma_decomp}
& & \overline{\Gamma}_{n}^{N}(\varphi)-\pi_{n}(
\varphi)
\nonumber
\\
& &\quad =\sum_{p=0}^{n} \bigl[\overline{
\Gamma}_{p,n}^{N}\overline {Q}_{p,n}(\varphi)-
\overline{\Gamma}_{p-1,n}^{N}\overline {Q}_{p-1,n}(
\varphi) \bigr]
\nonumber
\\
& &\quad =\sum_{p=0}^{n} \biggl[\sum
_{i_{p}}\beta_{p,n}^{i_{p}}\overline
{W}_{p}^{i_{p}}\overline{Q}_{p,n}(\varphi) \bigl(
\zeta_{p}^{i_{p}}\bigr)-\sum_{i_{p-1}}
\sum_{i_{p}}\beta_{p,n}^{i_{p}}\alpha
_{p-1}^{i_{p}i_{p-1}}\overline{W}_{p-1}^{i_{p-1}}\overline
{Q}_{p-1,n}(\varphi) \bigl(\zeta_{p-1}^{i_{p-1}}\bigr)
\biggr]
\\
& &\quad =\sum_{p=0}^{n}\sum
_{i_{p}}\beta_{p,n}^{i_{p}}\overline
{W}_{p}^{i_{p}} \biggl[\overline{Q}_{p,n}(\varphi)
\bigl(\zeta _{p}^{i_{p}}\bigr)-\frac{\sum_{i_{p-1}}\alpha_{p-1}^{i_{p}i_{p-1}}\overline
{W}_{p-1}^{i_{p-1}}\overline{Q}_{p-1,n}(\varphi)(\zeta
_{p-1}^{i_{p-1}})}{\overline{W}_{p}^{i_{p}}} \biggr]
\nonumber
\\
& &\quad =\sum_{p=0}^{n}\sum
_{i_{p}}\beta_{p,n}^{i_{p}}\overline
{W}_{p}^{i_{p}}\Delta_{p,n}^{i_{p}}=N^{-1/2}
\sum_{k=1}^{(n+1)N}\xi _{k}^{N}.\nonumber
\end{eqnarray}
Each $\xi_{k}^{N}$ is measurable w.r.t. $\mathcal{F}^{(k)}$ because,
using Corollary~\ref{cor:measurability}, \textup{(A2)} and (B)
we have that for any $k=1,\ldots,(n+1)N$, if we set $p:= \lfloor
(k-1)/N \rfloor$
and $i:=k-pN$, the quantity $\Delta_{p,n}^{i}$ is measurable
w.r.t. $\mathcal{F}^{(k)}$
and $\beta_{p,n}^{i_{p}}\overline{W}_{p}^{i_{p}}$ is measurable
w.r.t. $\mathcal{F}_{p-1}$.

To verify (\ref{eq:xi_cond_exp}), again use the fact that for any
$i\in[N]$ and $0\leq p\leq n$, $\beta_{p,n}^{i}\overline{W}_{p}^{i}$
is measurable w.r.t. $\mathcal{F}_{p-1}$, and note that given $\mathcal
{F}_{p-1}$,
the particles $ \{ \zeta_{p}^{i} \} _{i=1}^{N}$ are conditionally
independent, and distributed as specified in Algorithm \ref{alg:aSMC}.
Hence for any $k=1,\ldots,(n+1)N$ and $p:= \lfloor(k-1)/N
\rfloor$
and $i:=k-pN$, we have $\mathbb{E} [\xi_{k}^{N}|\mathcal
{F}^{(k-1)} ]=\sqrt{N}\beta_{p,n}^{i}\overline{W}_{p}^{i}\mathbb
{E} [\Delta_{p,n}^{i}|\mathcal{F}_{p-1} ]=0$.

For the inequality (\ref{eq:martingale_burkholder_bound}), by Minkowski's
inequality and (\ref{eq:Gamma-gamma_decomp}),
%
\begin{equation}
\mathbb{E} \bigl[\bigl\llvert \overline{\Gamma}_{n}^{N}(
\varphi)-\pi_{n}(\varphi )\bigr\rrvert ^{r}
\bigr]^{1/r}\leq\sum_{p=0}^{n}
\mathbb{E} \bigl[\bigl\llvert \overline{\Gamma}_{p,n}^{N}
\overline{Q}_{p,n}(\varphi)-\overline{\Gamma }_{p-1}^{N}
\overline{Q}_{p-1,n}(\varphi)\bigr\rrvert ^{r}
\bigr]^{1/r}.\label
{eq:Gamma-gamma_mink}
\end{equation}
For each term in (\ref{eq:Gamma-gamma_mink}), using the above stated
conditional independence and measurability properties, we may apply
Lemma~\ref{lem:del_moral_moment} to establish the existence of an independent
constant $B(r)$, depending only on $r$ and such that
%
\begin{eqnarray}\label{eq:GammaN-GammaN}
& & \mathbb{E} \bigl[\bigl|\overline{\Gamma}_{p,n}^{N}\overline
{Q}_{p,n}(\varphi)-\overline{\Gamma}_{p-1}^{N}
\overline {Q}_{p-1,n}(\varphi)\bigr|^{r}|\mathcal{F}_{p-1}
\bigr]
\nonumber
\\
& &\quad =\mathbb{E} \biggl[\biggl|\sum_{i}
\beta_{p,n}^{i}\overline {W}_{p}^{i}
\Delta_{p,n}^{i}\biggr|^{r}\Big|\mathcal{F}_{p-1}
\biggr]
\\
& & \quad\leq B(r)\operatorname{osc} \bigl(\overline{Q}_{p,n}(\varphi)
\bigr)^{r} \biggl(\sum_{i} \bigl(
\beta_{p,n}^{i}\overline{W}_{p}^{i}
\bigr)^{2} \biggr)^{r/2},\nonumber
\end{eqnarray}
almost surely. The proof is completed upon combining this estimate
with (\ref{eq:Gamma-gamma_mink}).
\end{pf*}

\begin{pf*}{Proof of Theorem~\ref{thmm:convergence}}
For part (1), note
\[
\overline{\Gamma}_{n}^{N}(1)-\pi_{n}(1)=
\frac{Z_{n}^{N}}{Z_{n}}-1,
\]
then applying Proposition~\ref{prop:martingale} with $\varphi=1$
and using (\ref{eq:Gamma_telescope})--(\ref{eq:xi_cond_exp}) gives
\[
\mathbb{E}\bigl[Z_{n}^{N}\bigr]=Z_{n}.
\]

Moving to the proof of part (2), let us assume for now,
only (A1), \textup{(A2)} and (B), but not necessarily (B$^{+}$).
Define $c_{n}:=\sup_{x}g_{n}(x)/\pi_{n}(g_{n})$. Under (A1),
we have
\[
\mathrm{osc} \bigl(\overline{Q}_{p,n}(\varphi) \bigr)\leq2\llVert
\varphi\rrVert \sup_{x}\overline{Q}_{p,n}(1) (x)
\leq2\llVert \varphi\rrVert \prod_{q=p}^{n-1}c_{q}<+
\infty
\]
and also using Lemma~\ref{lem:W_n_representation}, (\ref{eq:W_bar_defn})
and the fact that each $\alpha_{p}$ is a Markov transition matrix,
we obtain
\[
0<\overline{W}_{p}^{i_{p}}\leq\sum
_{ (i_{0},\ldots,i_{p-1} )\in
[N]^{p}}\prod_{q=0}^{p-1}c_{q}
\alpha_{q}^{i_{q+1}i_{q}}=\prod_{q=0}^{p-1}c_{q}<+
\infty.
\]
From (\ref{eq:martingale_burkholder_bound}), we then obtain
%
\begin{eqnarray}\label{eq:L_r_bound_proof_max}
& & \mathbb{E} \bigl[\bigl\llvert \overline{\Gamma}_{n}^{N}(
\varphi)-\pi _{n}(\varphi)\bigr\rrvert ^{r}
\bigr]^{1/r}
\nonumber
\\
& &\quad \leq2\llVert \varphi\rrVert B(r)^{1/r} \Biggl(\prod
_{p=0}^{n-1}c_{p} \Biggr)\sum
_{p=0}^{n}\biggl\llvert \sum
_{i} \bigl(\beta _{p,n}^{i}
\bigr)^{2}\biggr\rrvert ^{1/2}
\\
& &\quad \leq2\llVert \varphi\rrVert B(r)^{1/r} \Biggl(\prod
_{p=0}^{n-1}c_{p} \Biggr)\sum
_{p=0}^{n}\Bigl\llvert \max_{i\in[N]}
\beta _{p,n}^{i}\Bigr\rrvert ^{1/2},\nonumber
\end{eqnarray}
where the final inequality holds because $ \{ \beta_{p,n}^{i}
\} _{i\in[N]}$
is a probability vector. Then invoking (B$^{+}$), the convergence
in (\ref{eq:convergence_L_r_statement_gam_weak}) follows from (\ref
{eq:L_r_bound_proof_max})
applied with $\varphi=1$. For (\ref{eq:convergence_L_r_statement_pi_weak}),
we apply Minkowski's inequality, the fact $\llvert \Gamma_{n}^{N}(\varphi
)/\Gamma_{n}^{N}(1)\rrvert \leq\llVert \varphi\rrVert $
and (\ref{eq:L_r_bound_proof_max}) twice to obtain
%
\begin{eqnarray}\label{eq:L_r_bound_proof_pi}
\mathbb{E} \bigl[\bigl\llvert \pi_{n}^{N}(\varphi)-
\pi_{n}(\varphi)\bigr\rrvert ^{r} \bigr]^{1/r} &
\leq& \mathbb{E} \bigl[\bigl\llvert \overline{\Gamma }_{n}^{N}(
\varphi)-\pi_{n}(\varphi)\bigr\rrvert ^{r}
\bigr]^{1/r}
\nonumber
\\
&&{} +  \mathbb{E} \biggl[\biggl\llvert \frac{\Gamma_{n}^{N}(\varphi)}{\Gamma
_{n}^{N}(1)}\biggr\rrvert
^{r}\bigl\llvert \overline{\Gamma}_{n}^{N}(1)-1
\bigr\rrvert ^{r} \biggr]^{1/r}
\\
& \leq& 4\llVert \varphi\rrVert \bigl[B(r) \bigr]^{1/r} \Biggl(\prod
_{p=0}^{n-1}c_{p} \Biggr)\sum
_{p=0}^{n}\Bigl\llvert \max
_{i\in[N]}\beta _{p,n}^{i}\Bigr\rrvert
^{1/2}.\nonumber
\end{eqnarray}
The convergence in probability then follows from Markov's inequality,
completing the proof of part (2).

For part (3), under (B$^{++}$) we have $\beta_{p,n}^{i}=1/N$,
and therefore $\llvert \max_{i\in[N]}\beta_{p,n}^{i}\rrvert ^{1/2}=N^{-1/2}$.
Substituting this into (\ref{eq:L_r_bound_proof_max}) with $\varphi=1$,
and into (\ref{eq:L_r_bound_proof_pi}), gives (\ref
{eq:convergence_L_r_statement_gamm})--(\ref{eq:convergence_L_r_statement_pi}).
The almost sure convergence follows from the Borel--Cantelli lemma.
\end{pf*}

\subsection*{Proofs and auxiliary results for Section \protect\ref{sec:stability}}

\begin{pf*}{Proof of Proposition~\ref{prop:norm_const_bound}} The proof follows
a similar line of argument to \cite{DelMoral2013book}, Proof of Theorem~16.4.1,
but applies to a more general algorithm than considered there. To
start, we apply Proposition~\ref{prop:martingale}, equation (\ref
{eq:Gamma_telescope})
with $\varphi=1$ and (\ref{eq:xi_cond_exp}), we obtain
\[
\mathbb{E} \biggl[ \biggl(\frac{Z_{n}^{N}}{Z_{n}}-1 \biggr)^{2} \biggr]=
\sum_{p=0}^{n}\sum
_{i_{p}}\mathbb{E} \bigl[ \bigl(\beta_{p,n}^{i_{p}}
\overline {W}_{p}^{i_{p}}\Delta_{p,n}^{i_{p}}
\bigr)^{2} \bigr].
\]
Under (B$^{++}$) we have $\beta_{p,n}^{i_{p}}=1/N$, then
using the other hypotheses of the proposition and noting $\operatorname
{osc} (Q_{n,n}(1) )=\operatorname{osc} (1 )=0$,
we have for $n\geq1$,
\begin{eqnarray*}
\mathbb{E} \biggl[ \biggl(\frac{Z_{n}^{N}}{Z_{n}}-1 \biggr)^{2} \biggr] &
= & \sum_{p=0}^{n}\sum
_{i}\mathbb{E} \biggl[\frac{1}{N^{2}} \bigl(\overline
{W}_{p}^{i} \bigr)^{2} \bigl(
\Delta_{p,n}^{i} \bigr)^{2} \biggr]
\\
& \leq& \sum_{p=0}^{n-1}\operatorname{osc}
\bigl(\overline{Q}_{p,n}(1) \bigr)^{2}\mathbb{E} \biggl[
\frac{1}{N^{2}}\sum_{i} \bigl(\overline
{W}_{p}^{i} \bigr)^{2} \biggr]
\\
& = & \sum_{p=0}^{n-1}\operatorname{osc}
\bigl(\overline{Q}_{p,n}(1) \bigr)^{2}\frac{1}{N}
\mathbb{E} \biggl[\frac{1}{\mathcal{E}_{p}^{N}} \biggl(\frac
{1}{N}\sum
_{i}\overline{W}_{p}^{i}
\biggr)^{2} \biggr]
\\
& \leq& \sum_{p=0}^{n-1}\frac{\operatorname{osc} (\overline
{Q}_{p,n}(1) )^{2}}{N\tau_{p}}
\mathbb{E} \biggl[ \biggl(\frac{1}{N}\sum_{i}
\overline{W}_{p}^{i}-1 \biggr)^{2}+1 \biggr]
\\
& = & \sum_{p=0}^{n-1}\frac{\operatorname{osc} (\overline{Q}_{p,n}(1)
)^{2}}{N\tau_{p}}
\biggl(\mathbb{E} \biggl[ \biggl(\frac
{Z_{p}^{N}}{Z_{p}}-1 \biggr)^{2}
\biggr]+1 \biggr),
\end{eqnarray*}
where last two lines use $N^{-1}\sum_{i}\overline{W}_{p}^{i}=Z_{p}^{N}/Z_{p}$
and by Theorem~\ref{thmm:convergence}, $\mathbb{E} [Z_{p}^{N} ]=Z_{p}$.
\end{pf*}

\begin{pf*}{Proof of Proposition~\ref{prop:L_p_bound_mixing}}
First, note that
by the same arguments as in the proof of Proposition~\ref{prop:martingale},
equation (\ref{eq:GammaN-GammaN}), we have for any $\phi\in\mathcal{L}$,
$0\leq p\leq n$,
%
\begin{eqnarray}\label{eq:GammaN-GammaN_mixing_proof}
& & \mathbb{E} \bigl[\bigl|\Gamma_{p,n}^{N}Q_{p,n}(
\phi)-\Gamma _{p-1}^{N}Q_{p-1,n}(\phi)\bigr|^{r}|
\mathcal{F}_{p-1} \bigr]
\nonumber
\\[-8pt]
\\[-8pt]
\nonumber
& &\quad \leq B(r)\operatorname{osc} \bigl(Q_{p,n}(\phi)
\bigr)^{r} \biggl(\sum_{i} \bigl(
\beta_{p,n}^{i}W_{p}^{i}
\bigr)^{2} \biggr)^{r/2},
\end{eqnarray}
with the convention $\Gamma_{-1}^{N}Q_{-1,n}(\phi)=\gamma_{n}(\phi)$.

For the remainder of the proof, fix $\varphi\in\mathcal{L}$ arbitrarily,
and set $\bar{\varphi}:=\varphi-\pi_{n}(\varphi)$. Defining
\[
D_{p,n}^{N}:=\frac{\Gamma_{p,n}^{N}Q_{p,n} (\bar{\varphi}
)}{\Gamma_{p,n}^{N}Q_{p,n}(1)},\qquad 0\leq p\leq n,
\]
and then noting
\[
D_{n,n}^{N}=\frac{\Gamma_{n}^{N} (\bar{\varphi} )}{\Gamma
_{n}^{N}(1)}=\pi_{n}^{N}(
\varphi)-\pi_{n}(\varphi),
\]
we shall focus on the decomposition:
%
\begin{equation}
\pi_{n}^{N}(\varphi)-\pi_{n}(\varphi)  =
D_{0,n}^{N}+\sum_{p=1}^{n}D_{p,n}^{N}-D_{p-1,n}^{N},\label{eq:pi_n^N_decomp}
\end{equation}
with the convention that the summation is zero when $n=0$.

For $1\leq p\leq n$, write
\[
D_{p,n}^{N}-D_{p-1,n}^{N}=T_{p,n}^{ (N,1 )}+T_{p,n}^{
(N,2 )},
\]
where
\begin{eqnarray*}
T_{p,n}^{ (N,1 )} & := & \frac{1}{\Gamma
_{p,n}^{N}Q_{p,n}(1)} \bigl[
\Gamma_{p,n}^{N}Q_{p,n} (\bar{\varphi } )-
\Gamma_{p-1,n}^{N}Q_{p-1,n} (\bar{\varphi} ) \bigr],
\\
T_{p,n}^{ (N,2 )} & := & \frac{\Gamma_{p-1,n}^{N}Q_{p-1,n}(\bar
{\varphi})}{\Gamma_{p-1,n}^{N}Q_{p-1,n}(1)}\frac{ [\Gamma
_{p-1,n}^{N}Q_{p-1,n} (1 )-\Gamma_{p,n}^{N}Q_{p,n} (1
) ]}{\Gamma_{p,n}^{N}Q_{p,n} (1 )}.
\end{eqnarray*}

We have the estimates
%
\begin{equation}
\frac{\operatorname{osc} (Q_{p,n}(\phi) )}{\inf_{x}Q_{p,n} (1
)(x)}\leq2\delta_{p,n}\bigl\llVert P_{p,n}(
\phi)\bigr\rrVert \label
{eq:osc/inf_bound}
\end{equation}
(which is finite under assumption (C) -- see also (\ref
{eq:Q_p,n_bounded})),
and
%
\begin{equation}
\biggl\llvert \frac{\Gamma_{p-1,n}^{N}Q_{p-1,n}(\phi)}{\Gamma
_{p-1,n}^{N}Q_{p-1,n}(1)}\biggr\rrvert  \leq \bigl\llVert
P_{p-1,n}(\phi )\bigr\rrVert .\label{eq:Gamma/Gamma_bound}
\end{equation}
Applying (\ref{eq:GammaN-GammaN_mixing_proof}) with $\phi=\bar{\varphi}$,
using (\ref{eq:osc/inf_bound}) and noting that $\Gamma_{p,n}^{N}(1)$
is measurable w.r.t. $\mathcal{F}_{p-1}$, we obtain
\begin{eqnarray*}
\mathbb{E} \bigl[\bigl|T_{p,n}^{ (N,1 )}\bigr|^{r} |
\mathcal{F}_{p-1} \bigr]^{1/r} & \leq& B(r)^{1/r}2
\delta_{p,n}\frac
{\llVert  P_{p,n}(\bar{\varphi})\rrVert }{\Gamma
_{p,n}^{N}(1)} \biggl(\sum_{i}
\bigl(\beta_{p,n}^{i}W_{p}^{i}
\bigr)^{2} \biggr)^{1/2}.
\end{eqnarray*}
Applying (\ref{eq:GammaN-GammaN_mixing_proof}) with $\phi=1$, using
(\ref{eq:Gamma/Gamma_bound}) and the same measurability condition,
we obtain
\begin{eqnarray*}
\mathbb{E} \bigl[\bigl|T_{p,n}^{ (N,2 )}\bigr|^{r} |
\mathcal{F}_{p-1} \bigr]^{1/r}  \leq B(r)^{1/r}2
\delta_{p,n}\frac
{\llVert  P_{p-1,n}(\bar{\varphi})\rrVert }{\Gamma
_{p,n}^{N}(1)} \biggl(\sum_{i}
\bigl(\beta_{p,n}^{i}W_{p}^{i}
\bigr)^{2} \biggr)^{1/2}.
\end{eqnarray*}
Therefore, via Minkowski's inequality and using
\[
\bigl\llVert P_{p-1,n}(\bar{\varphi})\bigr\rrVert =\bigl\llVert
Q_{p-1,n}(\bar{\varphi})/Q_{p-1,n}(1)\bigr\rrVert =\bigl\llVert
Q_{p}Q_{p,n}(\bar{\varphi})/Q_{p-1,n}(1)\bigr\rrVert
\leq\bigl\llVert P_{p,n}(\bar{\varphi})\bigr\rrVert ,
\]
we have
%
\begin{equation}
\mathbb{E} \bigl[\bigl\llvert D_{p,n}^{N}-D_{p-1,n}^{N}
\bigr\rrvert ^{r} \bigr]^{1/r}\leq B(r)^{1/r}4
\delta_{p,n}\bigl\llVert P_{p,n}(\bar{\varphi })\bigr\rrVert
\mathbb{E} \bigl[\bigl\llvert \mathcal{C}_{p,n}^{N}\bigr
\rrvert ^{r} \bigr]^{1/r}.\label{eq:D_p-D_p_proof}
\end{equation}

For the remaining term, $D_{0,n}^{N}$, we have
\[
D_{0,n}^{N} = \frac{1}{\Gamma_{0,n}^{N}Q_{0,n}(1)} \bigl(\Gamma
_{0,n}^{N}Q_{0,n} (\bar{\varphi} )-
\gamma_{n}(\bar{\varphi }) \bigr),
\]
where the final equality holds since $\gamma_{n}(\bar{\varphi})=\gamma
_{n}(\varphi)-\gamma_{n}(1)\pi_{n}(\varphi)=0$.
Using (\ref{eq:GammaN-GammaN_mixing_proof}) and (\ref{eq:osc/inf_bound})
in a similar fashion to above, we obtain
%
\begin{equation}
\mathbb{E} \bigl[\bigl\llvert D_{0,n}^{N}\bigr\rrvert
^{r} \bigr]^{1/r}\leq B(r)^{1/r}2
\delta_{0,n}\bigl\llVert P_{0,n}(\bar{\varphi})\bigr\rrVert
\mathbb{E} \bigl[\bigl\llvert \mathcal{C}_{0,n}^{N}\bigr
\rrvert ^{r} \bigr]^{1/r}.\label{eq:D_0_proof}
\end{equation}
The proof is complete upon using Minkowski's inequality to bound the
moments of (\ref{eq:pi_n^N_decomp}) using (\ref{eq:D_p-D_p_proof})
and (\ref{eq:D_0_proof}).
\end{pf*}

\begin{lem}[(\cite{smc:the:CdMG11}, Corollary~5.2)]
\label{lem:var_bound} Suppose
that assumptions \textup{(A2)} and \textup{(B)} hold. If
%
\begin{equation}
\sup_{n\geq1}\mathbb{E} \biggl[ \biggl(\frac{Z_{n}^{N}}{Z_{n}}
\biggr)^{2} \biggr]^{1/n}\leq1+\frac{c_{1}}{N\tau},\label{eq:var_bound_lemma_hyp}
\end{equation}
then
\[
N\tau\geq nc_{1} \quad\Longrightarrow\quad \mathbb{E} \biggl[ \biggl(
\frac
{Z_{n}^{N}}{Z_{n}}-1 \biggr)^{2} \biggr]\leq\frac{2c_{1}n}{N\tau}.
\]
\end{lem}

\begin{pf}
Under \textup{(A2)} and (B), we have by Theorem~\ref{thmm:convergence} that
$\mathbb{E}[Z_{n}^{N}]=Z_{n}$. The hypothesis (\ref{eq:var_bound_lemma_hyp})
can then be stated equivalently as
\[
\mathbb{E} \biggl[ \biggl(\frac{Z_{n}^{N}}{Z_{n}}-1 \biggr)^{2} \biggr]
\leq \biggl(1+\frac{c_{1}}{N\tau} \biggr)^{n}-1\qquad \forall n\geq1.
\]
Using the fact that $\log(1+x)\leq x$ for any $x\geq0$ and $\mathrm{e}^{x}\leq1+2x$
for any $x\in[0,1]$, we conclude that
\[
\biggl(1+\frac{c_{1}n}{N\tau} \biggr)^{n}-1=\exp \biggl[n\log \biggl(1+
\frac
{c_{1}}{N\tau} \biggr) \biggr]-1\leq\exp \biggl(\frac{c_{1}n}{N\tau} \biggr)-1
\leq\frac{2c_{1}n}{N\tau}
\]
for any $N\tau\geq c_{1}n$.
\end{pf}

\subsection*{Proofs for Section \protect\ref{sec:Discussion}}
\begin{pf*}{Proof of Lemma~\ref{lem:complexity}} Label the vertices of the
graph corresponding to $A$ arbitrarily with the integers $[N]$.
Let $s\geq1$ be the number of connected components of this graph.
Then for each $\ell\in[s]$ let $B(\ell)$ be the set of labels of
the $\ell$th connected component. Since $A$ is a B-matrix, each
connected component is complete, so we have for any $\ell\in[s]$
and $i\in B(\ell)$,
%
\begin{equation}
W_{n}^{i}=\sum_{j}
\alpha_{n-1}^{ij}W_{n-1}^{j}g_{n-1}
\bigl(\zeta _{n-1}^{j}\bigr)=\bigl\llvert B(\ell)\bigr\rrvert
^{-1}\sum_{j\in B(\ell
)}W_{n-1}^{j}g_{n-1}
\bigl(\zeta_{n-1}^{j}\bigr).\label{eq:W_n_complexity}
\end{equation}
The complexity of calculating $W_{n}^{i}$ is thus $\mathrm{O}(\llvert B(\ell)\rrvert )$,
and since $W_{n}^{i}=W_{n}^{j}$ for all $i,j\in B(\ell)$, the complexity
of calculating $ \{ W_{n}^{i} \} _{i\in[N]}$ is $\mathrm{O} (\sum_{\ell\in[s]}\llvert B(\ell)\rrvert  )=\mathrm{O}(N)$.
Arguing similarly to (\ref{eq:W_n_complexity}), with $\alpha_{n-1}=A$
we find that under Algorithm \ref{alg:aSMC}, for each $\ell\in[s]$,
the $ \{ \zeta_{n}^{i} \} _{i\in B(\ell)}$ are conditionally
i.i.d. according to
%
\begin{equation}
\frac{\sum_{j\in B(\ell)}W_{n-1}^{j}g_{n-1}(\zeta_{n-1}^{j})f(\zeta
_{n-1}^{j},\cdot)}{\sum_{j\in B(\ell)}W_{n-1}^{j}g_{n-1}(\zeta
_{n-1}^{j})}.\label{eq:complexity_block_distbn}
\end{equation}
By the same arguments used in \cite{carpenter1999improved} to address
the BPF, drawing $\llvert B(\ell)\rrvert $ samples from (\ref
{eq:complexity_block_distbn})
can be achieved at $\mathrm{O}(\llvert B(\ell)\rrvert )$ complexity, and thus
the overall complexity of the sampling part of Algorithm \ref{alg:aSMC}
is $\mathrm{O}(\sum_{\ell\in[s]}\llvert B(\ell)\rrvert )=\mathrm{O}(N)$.
\end{pf*}

\begin{pf*}{Proof of Proposition~\ref{prop:Upon-termination-of}} We prove
(\ref{eq:W_k_explicit}) by induction. We have
\[
\mathbb{W}_{0}^{i}=W_{n-1}^{i}g_{n-1}
\bigl(\zeta_{n-1}^{i}\bigr)=\sum_{j\in
B(0,i)}W_{n-1}^{j}g_{n-1}
\bigl(\zeta_{n-1}^{j}\bigr)
\]
and, when (\ref{eq:W_k_explicit}) holds at rank $k$, we have at
rank $k+1$,
\begin{eqnarray*}
\mathbb{W}_{k+1}^{i} & = & \mathbb{W}_{k}^{\mathcal
{I}_{k}(2i-1)}/2+
\mathbb{W}_{k}^{\mathcal{I}_{k}(2i)}/2
\\
& = & 2^{-(k+1)}\sum_{j\in B(k,\mathcal
{I}_{k}(2i-1))}W_{n-1}^{j}g_{n-1}
\bigl(\zeta_{n-1}^{j}\bigr)
\\
&&{} +  2^{-(k+1)}\sum_{j\in B(k,\mathcal
{I}_{k}(2i))}W_{n-1}^{j}g_{n-1}
\bigl(\zeta_{n-1}^{j}\bigr)
\\
& = & 2^{-(k+1)}\sum_{j\in B(k,\mathcal{I}_{k}(2i-1))\cup B(k,\mathcal
{I}_{k}(2i))}W_{n-1}^{j}g_{n-1}
\bigl(\zeta_{n-1}^{j}\bigr)
\\
& = & 2^{-(k+1)}\sum_{j\in B(k+1,i)}W_{n-1}^{j}g_{n-1}
\bigl(\zeta_{n-1}^{j}\bigr).
\end{eqnarray*}
Finally, for any $i\in[N/2^{k}]$ and $j\in B(k,i)$
\[
W_{n}^{j}=\sum_{\ell}
\alpha_{n-1}^{i\ell}W_{n-1}^{\ell}g_{n-1}
\bigl(\zeta _{n-1}^{\ell}\bigr)=2^{-k}\sum
_{\ell\in B(k,i)}W_{n-1}^{\ell}g_{n-1}\bigl(
\zeta _{n-1}^{\ell}\bigr)=\mathbb{W}_{k}^{i},
\]
which establishes (\ref{eq:W_k_explicit})--(\ref{eq:W_equals_bb_W}).

No matter what adaptation rule of Table~\ref{tab:Choosing_I_k} is
used, each quantity $ \{ B(k,i) \} _{i\in[N/2^{k}]}$ obtained
by Algorithm \ref{alg:generic adaptation} is, by construction, a
partition of $[N]$ and thus the $\alpha_{n-1}$ output by Algorithm
\ref{alg:generic adaptation} is a B-matrix. Noting that a B-matrix
always admits the uniform distribution on $[N]$ as an invariant distribution,
we have for any B-matrix, say $A$, the identity $\overline{\mathbb
{W}}_{0}=N^{-1}\sum_{i}\sum_{j}A^{ij}W_{n-1}^{i}g_{n-1}(\zeta_{n-1}^{i})$
and so upon termination of the ``while'' loop in Algorithm \ref
{alg:generic adaptation},
$\mathcal{E}=\mathcal{E}_{n}^{N}$ and hence $\mathcal{E}_{n}^{N}\geq\tau$
always.

For the Simple and Random adaptation rules, the worst case complexity
of Algorithm \ref{alg:generic adaptation} is as follows. The part
of the algorithm preceding the ``while'' loop is $\mathrm{O}(N)$. The complexity
of iteration $k$ of the ``while'' loop is $\mathrm{O}(N/2^{k})$, the worst
case is when the loop terminates with $k=m$, in which case the complexity
of the ``while'' loop is $\mathrm{O}(\sum_{k=0}^{m}N/2^{k})$, thus the overall
complexity is no more than $\mathrm{O}(N)$.

For the Greedy procedure, the sort operation required to obtain
$\mathcal{I}_{k}$
is of complexity $\mathrm{O} (N/2^{k}\log_{2} (N/2^{k} ) )$,
and so in the worst case, the complexity of the ``while'' loop is
of the order
\[
t(N):=\sum_{k=0}^{m}\frac{N}{2^{k}}
\log_{2} \biggl(\frac{N}{2^{k}} \biggr),
\]
or expressed recursively, $t(N)=t(N/2)+N\log_{2}N$, and $t(2)=2$.
A simple induction shows that this recursion has solution
$t(N)=2[1+N(\log_{2}N-1)]$,
hence the overall worst case complexity of the ``while'' loop is
$\mathrm{O}(N\log_{2}N)$. The proof is complete since by Lemma~\ref{lem:complexity},
the complexity of operations in Algorithm \ref{alg:aSMC} other than
line $(\star)$ is $\mathrm{O}(N)$.
\end{pf*}
\end{appendix}

\section*{Acknowledgements}

The authors would like to thank the Associate Editor and a referee
for helpful comments and corrections. The first and third authors
were supported by EPSRC grant EP/K023330/1.



%




\printhistory
\end{document}